%% file: paper.tex
\newif\iffocs
\focsfalse

\iffocs
\documentclass[conference,final]{IEEEtran}
\IEEEoverridecommandlockouts
\else
\documentclass[letterpaper, 11pt]{article}
\fi

\input{macros}
\date{}
\title{Instance-Optimality in
  I/O-Efficient Sampling and Sequential Estimation\thanks{Jakub Tětek and Mikkel Thorup were supported by the VILLUM Foundation grant 16582. Shyam Narayanan was supported by the NSF Graduate Research Fellowship, and a Google Fellowship. Jakub Tětek and Václav Rozhoň were supported by the Ministry of Education and Science of Bulgaria (support for INSAIT, part of the Bulgarian National Roadmap for Research Infrastructure).\\ The work was done while Jakub Tětek was affiliated with BARC, University of Copenhagen. The work was partially done while Václav Rozhoň was visiting MIT.}
}

\iffocs
\author{
\IEEEauthorblockN{Shyam Narayanan}
\IEEEauthorblockA{
\small MIT\\
\small \texttt{shyamsn@mit.edu}\\
}
\and
\IEEEauthorblockN{Václav Rozhoň}
\IEEEauthorblockA{
\small INSAIT, Sofia University\\``St. Kliment Ohridski''\\
\small \texttt{vaclavrozhon@gmail.com}\\
}
\and
\IEEEauthorblockN{Jakub Tětek}
\IEEEauthorblockA{
\small INSAIT, Sofia University\\``St. Kliment Ohridski''\\
\small \texttt{j.tetek@gmail.com}\\
}
\and
\IEEEauthorblockN{Mikkel Thorup}
\IEEEauthorblockA{
\small BARC, University of Copenhagen\\
\small \texttt{mikkel2thorup@gmail.com}\\
}
}
\else

\author{
Shyam Narayanan\\
\texttt{shyamsn@mit.edu}\\
MIT
\and
    Václav Rozhoň\\
    \texttt{vaclavrozhon@gmail.com}\\
    INSAIT, Sofia University\\ ``St. Kliment Ohridski''
    \and
	Jakub Tětek\\
	\texttt{j.tetek@gmail.com}\\
	INSAIT, Sofia University\\ ``St. Kliment Ohridski''
    \and
        Mikkel Thorup\\
        \texttt{mikkel2thorup@gmail.com}\\
        BARC, Univ. of Copenhagen
}
\fi

\begin{document}

%
\maketitle

\begin{abstract}
Suppose we have a memory storing $0$s and $1$s and we want to estimate the frequency of $1$s by sampling. We want to do this I/O-efficiently, exploiting that each read gives a block of $B$ bits at unit cost; not just one bit.  If the input consists of uniform blocks: either all 1s or all 0s, then sampling a whole block at a time does not reduce the number of samples needed for estimation. On the other hand, if bits are randomly permuted, then getting a block of $B$ bits is as good as getting $B$ independent bit samples.  However, we do not want to make any such assumptions on the input. Instead, our goal is to have an algorithm with instance-dependent performance guarantees which stops sampling blocks as soon as we know that we have a probabilistically reliable estimate. We prove our algorithms to be instance-optimal among algorithms oblivious to the order of the blocks, which we argue is the strongest form of instance optimality we can hope for. We also present similar results for I/O-efficiently estimating mean with both additive and multiplicative error, estimating histograms, quantiles, as well as the empirical cumulative distribution function.

We obtain our above results on I/O-efficient sampling by reducing to corresponding problems in the so-called sequential estimation. In this setting, one samples from an unknown distribution until one can provide an estimate with some desired error probability. Sequential estimation has been considered extensively in statistics over the past century. However, the focus has been mostly on parametric estimation, making stringent assumptions on the distribution of the input, and thus not useful for our reduction. In this paper, we make no assumptions on the input distribution (apart from its support being a bounded set). Namely, we provide non-parametric instance-optimal results for several fundamental problems: mean and quantile estimation, as well as learning mixture distributions with respect to $\ell_\infty$ and the so-called Kolmogorov-Smirnov distance.

All our algorithms are simple, natural, and practical, and some are even known from other contexts, e.g., from statistics in the parameterized setting. The main technical difficulty is in analyzing them and proving that they are instance optimal.
\end{abstract}

\vasek{Vasek to jakub -- I would start the abstract with something like this (see below). this is intended just as a food for thought}

\iffocs
\else
 \thispagestyle{empty}

 \thispagestyle{empty}

%
%

\newpage
 \thispagestyle{empty}
\tableofcontents
\newpage

\clearpage
\setcounter{page}{1}
\fi


%
%

\section{Introduction}
Suppose we have computer memory that stores a sequence of $0$s and $1$s. How do we best estimate the frequency of 1s in this memory? We would like to exploit the fact that every time we access the memory, we get a whole cache-line of $B$ bits (and not just one bit). Unfortunately, not much can be said about this problem if one adopts the lens of worst-case analysis: After all, on inputs where each block contains only zeros or only ones, sampling an entire block offers no advantage over sampling a single bit. 
However, while we cannot do better \emph{in general}, we can do much better almost always. We would like to have algorithms with beyond-worst-case theoretical guarantees, but we also would like to not make assumptions on what the input looks like. To do this, we use the notion of \emph{instance-optimality}. We use this notion to investigate the above counting problem, as well as approximately computing quantiles, histograms, and empirical cumulative distribution functions in the I/O setting. Using instance optimality, we prove strong beyond-worst-case theoretical guarantees about our algorithms, even when the problems are not interesting from the worst-case perspective.

One of the main foci of this paper is showing that we may solve all the mentioned problems by reducing them to problems in the statistical setting of \emph{sequential estimation}. 
In this setting, the goal is to create adaptive estimators/algorithms that make as few samples as possible, stopping as soon as a desired level of accuracy is reached.
We obtain all our results on I/O-efficient computation by reducing the problems to the sequential estimation setting and then solving the respective problems that we reduced to.
%
There are thus two closely related parallel stories running through this paper: that of sequential estimation and that of estimation in the external memory setting. 

For example, our bit-counting problem reduces to the following problem: Given a distribution on $[0,1]$, estimate the mean to within additive $\epsilon$ with probability, say, $2/3$.
The standard solution is to make $\Theta(1/\epsilon^2)$ samples and use the sample mean as the estimate. 
While this algorithm is optimal from the standard, worst-case, perspective, in this paper we challenge the assumption that the worst-case perspective is the ``right'' one for this kind of problems. In fact, one can often do significantly better, with $O(1/\epsilon)$ samples sufficing in some very natural situations. 

We investigate this and other estimation problems and present algorithms that are optimal in the sense of \emph{instance-optimality}, meaning that \emph{we are as efficient as possible on any instance}  (see \Cref{def:instance_optimality}), inspired by the work of \citet{dagum2000optimal} who gave an analogous result for mean estimation with \emph{relative} error. Namely, we devise instance-optimal algorithms for mean estimation (with additive error, in contrast to \cite{dagum2000optimal}), as well as for estimating quantiles and learning mixture distributions, with respect to $\ell_\infty$ and the Kolmogorov-Smirnov\footnote{The Kolmogorov-Smirnov distance is defined as $KS(X_1,X_2) = \sup_x |P[X_1 \leq x] - P[X_2 \leq x]|$.} (KS) error metrics. We also show that our algorithm for additive mean estimation can be used by a simple reduction to solve the relative version of the problem.

All our algorithms are very simple, with the main technical difficulty being in proving that they are instance optimal. For example, for the simplest problem of estimating the mean of a distribution on $[0,1]$, we show that the following algorithm is instance optimal:

\vspace{1em}
\algorithmstyle{plain}
\begin{minipage}{\linewidth}
\hspace{.3cm} 
\begin{algorithm}[H]
\label{alg:mean_informal}
$S_1 \leftarrow$ Get $O(1/\epsilon)$ samples\\
Let $\tilde{\sigma}^2$ be the empirical variance of $S_1$\\
$S_2 \leftarrow$ Get $O(1/\epsilon + \tilde{\sigma}^2/\epsilon^2)$ samples\\
Return the sample mean of $S_2$
\end{algorithm}
\end{minipage}
\vspace{1em}
\algorithmstyle{ruled}

Interestingly, a very similar two-phase algorithm has already appeared 70 years ago in the statistics literature, namely in the works of \citet{stein1945two} and \citet{cox1952estimation}. However, these papers work in a significantly different (parametric) setting and their analysis of the algorithm is fundamentally different. 
We instead use the framework of instance optimality to analyze this algorithm. We stress that instance optimality is essentially the strongest guarantee that one can prove about an algorithm: it states that \emph{no correct algorithm} can beat the performance of the algorithm at hand \emph{on any possible input} by more than a constant factor. 

\paragraph{Additive errors.}
We note that the above simple algorithm only works for additive
errors which are the main focus of this paper. Multiplicative errors were considered in  \citet{dagum2000optimal}; we show in \Cref{sec:mean_multiplicative} that our algorithm for additive error can be easily modified to achieve also the same guarantee as \citet{dagum2000optimal}. Apart from mean estimation, we also extend our basic 
additive algorithm to the estimation of medians, quantiles, and learning distributions in $\ell_\infty$ and KS distance.  

%

\paragraph{Sequential estimation.} 
Sequential estimation in statistics focuses on developing algorithms that stop sampling once the collected data suffices to provide a reliable estimate, often requiring fewer samples than in the worst case.
This area of statistics has seen a lot of interest over the past century \cite{sequential_survey}. 
However, almost all existing approaches fall short of being instance-optimal, being either overly reliant on specific assumptions about sample distributions (i.e., they are parametric) or lacking in proven optimality. A more detailed comparison is presented in \Cref{sec:related_work}. 

Using techniques from theoretical computer science, more specifically from distribution testing
, we get (nearly) instance-optimal sequential estimators 
for the following problems: mean estimation (both with additive and multiplicative error), quantile estimation, and learning mixture distributions w.r.t.\ the $\ell_\infty$ and Kolmogorov-Smirnov (KS) distance. 

We stress that our algorithms make no assumptions on the input distribution.
This is important, because without this, we cannot make the reductions to/from sequential estimation work. Indeed, this is why the known results in sequential estimation are not sufficient for our purposes.

\paragraph{I/O-efficient sampling.} One application of our results on sequential estimation, and indeed the original motivation for considering the above setting, is I/O-efficient sampling. 

Suppose we have block-based access to the memory. Namely, assume we are in the standard external memory setting\footnote{We ignore the size of the cache as it is not relevant to our applications. Instead, our focus is on the block size $B$.}: whenever we perform a read operation, we get a block of memory of size $B$ at a unit cost. 
Our goal is to get as much benefit from seeing the whole block, as possible. 

Notice that for many problems, accessing an entire block of $B$ items does not help \emph{in the worst case}. 
For example, in the case of estimating the number of $1$'s in a $0/1$ string to within additive $\pm \eps$, if each block consists of only $1$'s or only $0$'s, we get no benefit from getting the whole block and we need $\Theta(1/\epsilon^2)$ samples. This is the same number as if the block size is $1$.

On the other hand, if the blocks look random-like, the fraction of $1$'s in a block is roughly the same as the global fraction of $1$'s. 
If the blocks are large enough, it turns out that for random-like blocks it suffices to take a much smaller sample of $O(1/\epsilon)$ blocks to produce a correct estimate. 
These many blocks are also necessary as we have to check the possibility of $\eps$ fraction of blocks having a significantly different fraction of $1$'s than the rest. The much better sample complexity of $O(1/\epsilon)$ on this instance is not revealed if we view the problem in the framework of worst-case complexity.


The standard remedy to this issue is to use parameterized complexity \cite{cygan2015parameterized}, i.e., to find a reasonable parameter of an input that distinguishes the hard and the easy inputs above. The general problem with parameterized complexity is that one must separately come up in an ad-hoc fashion with a parameter for each considered problem and it is questionable whether the parameter captures practical instances. 

In the spirit of this paper, we approach this problem with the lens of instance optimality, i.e., we want to use parametrization where the parameter is the input itself. 
However, as we also prove in \Cref{sec:model_motivation}, the standard notion of instance optimality cannot be achieved due to silly counterexamples that we cover later. 
Thus, we instead use a standard variant of instance optimality known as the \emph{order-oblivious instance optimality}\cite{afshani2017instance} where we ignore the order of the blocks in the input. 
We provide more motivation behind this setting and justify it in detail in \Cref{sec:model_motivation}. 

Fortunately, it turns out that there is a strong link between the models of sequential estimation and I/O efficient sampling algorithms. 
We later prove \Cref{thm:io_transfer_informal} that claims under mild technical assumptions the equivalence of instance optimality in the two settings. This provides a general way of transferring instance optimal sequential estimation algorithms to order-oblivious instance optimal I/O efficient algorithms, and vice versa. 

This equivalence between the two settings motivates the general strategy of this paper: We provide (nearly) instance optimal algorithms in the model of sequential estimation for a number of problems: approximate counting, estimating the average, quantiles, histograms, or the empirical distribution. Using the equivalence result between the models, we get I/O-efficient algorithms for all these problems that satisfy a very strong notion of (nearly) order-oblivious instance optimality. 
In other words, our I/O-efficient algorithms get essentially all possible benefit of sampling long blocks from memory that any algorithm could get.




\paragraph{Experimental evaluation: our approach is practical.} While the focus of this paper is the theoretical understanding of the two investigated settings and their relation, we present a brief experiment in \Cref{sec:experiments}, demonstrating the practical value of our work.


\begin{figure*}[t]
  \centering
  \includegraphics[width=\linewidth]{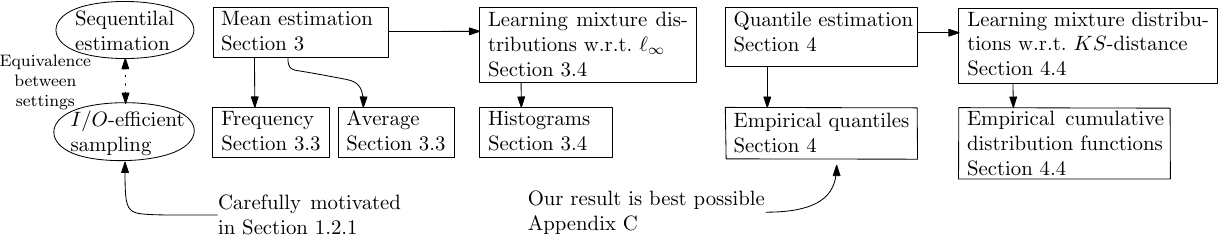}
  \caption{Our results and the implications between them. We use the equivalence between the I/O and sequential estimation settings (see \Cref{thm:io_transfer}) to get the implications going from the top (in the sequential estimation setting) to the bottom (on I/O-efficient sampling).}
  \label{fig:paper_structure}
\end{figure*}

\paragraph{How to read this paper.}
In \Cref{sec:intro_results}, we discuss our results as well as our techniques. In \Cref{sec:model_motivation}, we discuss I/O-efficient computation and we motivate the notion of order-oblivious instance optimality. 

\Cref{sec:preliminaries} is preliminaries, and we suggest that on the first reading, the reader only skims through it. In \Cref{sec:means}, we discuss mean estimation, counting, distribution learning (w.r.t.\ $\ell_\infty$), and histograms in both distributional and I/O settings. In \Cref{sec:quantiles}, we discuss quantile estimation and distribution learning (w.r.t.\ KS distance), again in both settings. In \Cref{sec:experiments}, we perform a basic experiment, observing the practicality of our approach.


\subsection{Our results and techniques} \label{sec:intro_results}
In this section, we discuss our results on sequential estimation and I/O-efficient sampling, as well as the techniques that we use to get our results. Throughout this section, we use $\eps$ and $\delta$ to denote the additive error of an algorithm and the algorithm's failure probability, respectively.


\subsubsection{Results on sequential estimation}
Suppose that we are given sampling access to an unknown distribution $D$ and we want to estimate a parameter $\theta(D)$. In the sequential estimation setting, we have a desired bound $\eps$ on the error and seek to achieve this error with at least a desired probability $1-\delta$, while minimizing the (expected) sample complexity of the algorithm. We discuss in \Cref{sec:related_work} the statistical literature on sequential estimation, as well as why most of those results cannot be easily turned into efficient algorithms.

We start by investigating one of the simplest such problems, mean estimation. We prove the following result:
\begin{result}
\label{res:mean}
There is an instance-optimal sequential estimator for estimating the mean of a random variable with a bounded range, with additive error.
\end{result}
We also show in \Cref{sec:mean_multiplicative} that this can be extended to mean estimation with multiplicative error, thus obtaining an alternative proof of the result of \cite{dang2023optimalitymeanestimationworstcase}.
We also show how one can estimate the quantile of a distribution, constructing a nearly instance-optimal algorithm. This is the most technically challenging part of this paper, especially the proof of the optimality of the factor that we lose which can be found in \Cref{sec:gap}. 

\begin{result}
\label{res:quantile}
We give an instance optimal, up to a factor of $O\left( 
1+\frac{\log \left( \delta^{-1} \log \eps^{-1} \right)}{\log \delta^{-1}}
\right)$,  sequential estimator for estimating the $q$-th quantile of a random variable. Moreover, this factor is the smallest factor possible. 
%
\end{result}
The error $\eps$ is here measured as the difference between $q$ and the closest value $q'$ for which the returned value is the $q'$-th quantile \emph{exactly}. This is analogous to the error measure standard in, for example, streaming algorithms for quantiles.


Using these two results, we can use standard reductions to construct nearly instance-optimal algorithms for learning distributions with respect to two different metrics. Namely, we consider learning a distribution w.r.t.\ $\ell_\infty$ and the Kolmogorov-Smirnov distance. 

\begin{result}
\label{res:dist}
We give instance-optimal, up to a factor of $O\left( 
\log\log\eps^{-1} \cdot \left( 1+\frac{\log \eps^{-1}}{\log \delta^{-1}}\right)
\right)$, sequential estimators for learning distributions w.r.t.\ the $\ell_\infty$ distance and the Kolmogorov-Smirnov distance.
%
\end{result}

In fact, we solve problems that are a bit more general than those discussed in \Cref{res:quantile,res:dist}. Instead of solving the estimation/learning problems on a distribution we have a sampling access to, we solve the problems on a \emph{mixture distribution}. Here, instead of having access to a distribution over some set $S$, we have a distribution over distributions on $S$. 
%
Our task is to solve the task (e.g. estimate the $q$-th quantile) on the mixture distribution $M(D)$ defined as the ``weighted average'' of sampled distributions (see \Cref{def:mixture}). 
We solve this more abstract problem mostly because we need this for our order-oblivious I/O-efficient algorithms. We discuss details in \Cref{sec:techniques}. 

\subsubsection{Results on I/O-efficient estimation}
\label{sec:results_io}
We now discuss our results on $I/O$-efficient estimation. Our results rely on the notion of \emph{order-oblivious instance-optimality} \cite{afshani2017instance}, which we introduce and motivate in \Cref{sec:model_motivation}. Roughly speaking, an algorithm satisfying this notion is instance-optimal if the input is viewed as a multiset. 
We need to use this weaker notion of optimality, as there are silly counterexamples that prevent us from achieving instance optimality.


We observe that there is a strong connection between the sequential estimation setting and the setting of I/O-efficient computation. 
For example, consider the problem of estimating the mean in the I/O model where each read retrieves $B$ consecutive values. This problem is suspiciously similar to the problem of finding a sequential estimator for the block-mean problem where the input distribution is over blocks of $B$ numbers and we want to estimate the overall mean. 

We prove that, up to some technical assumptions, the two settings of sequential estimation and I/O-efficient computation are equivalent. We discuss this result later in \Cref{sec:techniques}. Combining this with slightly more general versions of \Cref{res:mean,res:quantile,res:dist}, we get the following result with minimal additional effort:

\begin{result}
\label{res:io}
We give (near) order-oblivious instance-optimal algorithms for the following problems: I/O-efficient estimation of frequencies, means, histograms, empirical quantiles, and empirical cumulative distribution functions.   
\end{result}





\subsubsection{Our techniques}
\label{sec:techniques}

We start by sketching the techniques that go into proving our simplest result: an optimal sequential estimator for the mean of a distribution on a bounded range (\Cref{res:mean}). We only sketch here the upper bound -- while we give a proof of the lower bound for the sake of completeness in this paper, a proof appeared in \cite{dagum2000optimal,dang2023optimalitymeanestimationworstcase} and we do not give the intuition for the proof here.
We then go on to discuss the problem of estimating quantiles. Finally, we discuss our general equivalence result between the sequential estimation setting and I/O-efficient estimation.

\paragraph{Mean estimation upper bound (\Cref{res:mean}).} 
We have access to independent samples $X_1, X_2, \cdots$ from a distribution supported on $[0,1]$ and we want to estimate its mean.
We use the algorithm from page \pageref{alg:mean_informal}. Recall that it has two phases. 
In the first phase, we use $T_1 = O(1/\eps)$ samples to get an estimate $\tilde{\sigma}^2$ of the variance of the distribution. 
In the second phase, we use $T_2 = O(1/\eps + \tilde{\sigma}^2/\eps^2)$ samples to estimate the mean. 
We will prove that the algorithm returns a correct estimate with constant probability and its expected sample complexity is 
\begin{align}
\label{eq:complexity}
O\left(\frac{1}{\eps} + \frac{\sigma^2}{\eps^2}\right).     
\end{align}

Note that if we used $T_2' = \sigma^2/\eps^2$ samples in the second phase instead of $T_2$, we could use Chebyshev's inequality to conclude that with constant probability, the final estimate satisfies $\hat\mu = \mu \pm O\left( \sigma / \sqrt{T_2'} \right) = \mu \pm O(\eps)$, as needed. Such an algorithm would also have the desired sample complexity. The issue of course is that we do not know the value of $\sigma^2$ and we thus have to use $T_2$ samples instead of $T_2'$. 

However, we can observe that our estimate $\tilde\sigma^2$ of $\sigma^2$ is unbiased. This implies that the expected sample complexity of our algorithm is equal to \eqref{eq:complexity}. Furthermore, we are next going to show that with constant probability, $T_2 = \Omega(T_2')$. In particular, we next prove that with constant probability,  
\begin{align}
    \label{eq:want}
  1/\eps + \tilde{\sigma}^2/\eps^2 \ge \frac12 \cdot \sigma^2/\eps^2.   
\end{align}
Whenever we have $\sigma^2 \le \eps$, we can rewrite this condition as $\sigma^2/\eps^2 \le 1/\eps$ and conclude that \cref{eq:want} holds. We thus next assume that $\sigma^2 \ge \eps$. 
Recall that $\tilde\sigma^2$ is the empirical variance of the samples $X_1, \cdots, X_{T_1}$ defined as 
\begin{align}
\label{eq:slit}
\tilde\sigma^2 = \frac{1}{T_1 - 1} \sum_{i = 1}^{T_1} (X_i - \tilde\mu)^2
\end{align}
Routine calculations justify that we may proceed as if the right-hand side of \cref{eq:slit} operated with the true mean, $\mu$, instead of its estimate $\tilde\mu$; in particular, we can use Chebyshev's inequality to obtain that with constant probability, we have



$$
\tilde\sigma^2 = \sigma^2 \pm O\left( \sqrt{\frac{\Var\left[ (X_1-\mu)^2 \right]}{T_1}}\right).
$$ 
Using the fact that our distribution is supported on $[0,1]$, we can compute that $\Var[(X_1-\mu)^2] \le \Var[(X_1-\mu)] = \sigma^2$. Using this bound and the assumption $\sigma^2 \ge \eps$, we conclude that with constant probability, 
\begin{align}
\label{eq:bore}
\tilde\sigma^2 = \sigma^2 \pm O\left( \frac{\sigma}{\sqrt{1/\eps}}\right). 
\end{align}
Moreover, by choosing $T_1 = O(1/\eps)$ large enough we make sure that the constant factor in the standard deviation of $\tilde\sigma^2$ in \cref{eq:bore} is small enough to let us conclude that with constant probability, $\tilde\sigma^2 \ge \sigma^2)$. \Cref{eq:want} is thus proven. 

\medskip \noindent
\emph{Parametrization vs. instance-optimality.} As a small aside, one can ask why we are thinking of this algorithm as an instance-optimal algorithm and not as ``an optimal algorithm with respect to the parametrization of input by its variance''. First, notice that since our lower bound presented next holds for every distribution, our analysis additionally shows that parametrizing the sample complexity by the variance is the best possible parametrization. Moreover, in other problems we consider, it is less clear at first glance whether the parameter we consider is interesting. Thus it is preferable to think of the algorithm from page \pageref{alg:mean_informal} as being instance-optimal, and not just ``optimal if we parametrize by variance''.

\paragraph{Quantile estimation upper bound (\cref{res:quantile})} The above-described upper and lower bounds are for the simple case of mean estimation with additive error. Some more complications arise with estimating quantiles. 
We will next present the intuition for estimating quantiles of mixture distributions. Recall that in mixture distributions, each sample is itself a distribution. 

Our plan to estimate the $q$-th quantile of mixture distribution is to use an algorithm very similar to the algorithm from page \pageref{alg:mean_informal}. However, it is not clear what the right analog of the variance $\sigma^2$ and its estimate $\tilde\sigma^2$ should be in this case. 

Here is the right parameter: First, let $N \in [0, 1]$ be the number at the $q$-th quantile in the input distribution $D$. Consider the quantity 
\[
\sigma_q^2 = \Var_{X \sim D}[q_X(N)]
\]
where $q_X(N)$ stands for the quantile of the number $N$ in the distribution $X$ sampled from $D$. The parameter $\sigma_q^2$ is the right analog of the variance $\sigma^2$ from the mean estimation problem. 
\vasek{check notation}

To approximate $\sigma_q^2$, we can sample distributions from $D$, compute the empirical $q$-th quantile $\tilde{N}$, and then compute the empirical estimate $\tilde\sigma_q^2$ analogously to how we obtain the empirical estimate $\tilde\sigma^2$ in the mean estimation scenario.


Up to the fact that we need the above definition $\sigma_q^2$ and $\tilde\sigma_q^2$ instead of $\sigma^2$ and $\tilde\sigma^2$, our basic algorithm for quantile estimation is analogous to our mean estimation algorithm. 

While this approach results in an instance-optimal sequential estimator in the case of continuous distributions, the case of discrete distributions hides an additional difficulty that we need to address. Consider the following example: Suppose that the $10\%$ and $90\%$ quantiles of the input distribution $D$ are equal. We claim that the median of $D$ can be estimated up to arbitrarily small error $\eps$ in the number of samples that is \emph{not dependent on $\eps$}. Indeed, even if $\eps$ is tiny, we can learn, say, the $25\%$ and $75\%$ quantiles of $D$ up to an error of $15\%$. Then, we can observe that they are equal, which implies that the obtained estimate is also equal \emph{exactly} to the $50\%$ quantile of $D$. 

By first checking for such cases, we manage to improve our basic quantile estimation algorithm on some distributions. Note that this check has to be performed on ``all scales", meaning that our final algorithm performs this check iteratively for smaller and smaller values of $\bar{\eps}$: The parameter $\bar{\eps}$ ranges from constant to the desired accuracy of $\eps$.
The need to consider all the different $O(\log \eps^{-1})$ scales and the fact that we have to union bound the error probability over them, is why the sample complexity of our final algorithm incurs an additional multiplicative factor of size roughly $\log\log\eps^{-1}$. 

\paragraph{Quantile estimation lower bound (\cref{res:quantile})}
Our lower bound for the quantile estimation problem relies in a black-box fashion on the mean estimation lower bound.
Specifically, for any input distribution $D$, we use that lower bound to construct $D_1$ that is hard to distinguish from $D$ such that the $q$-th quantile $N$ of $D$ is the $(q-2 \eps)$-th quantile of $D_1$. 
Similarly, we construct $D_2$ such that $N$ will be $(q+2\eps)$-th quantile of $D_2$. Assuming the original distribution $D$ is continuous, the sets of correct answers for $D_1$ and $D_2$ are then disjoint, as needed. If the distribution $D$ is not continuous, additional care is needed. 

The additional multiplicative factor of roughly $\log\log\eps^{-1}$ that we needed in our quantile estimation algorithm does not appear in our lower bound. That is, our quantile estimation algorithm is instance-optimal only up to a factor of roughly $\log\log\eps^{-1}$. 
This is not just a weakness of our analysis: in \cref{sec:gap}, we show that this small gap between the upper and lower bound cannot be improved in general. 
In particular, we carefully analyze the problem of estimating the median of distributions of the form $\bar{D}(1/2 + \bar{\eps})$: This is a distribution parameterized by $\bar{\eps}$ that samples $0$ with probability $1/2 + \bar{\eps}$ and $1$ otherwise. We show that for a constant $\delta > 0$, sufficiently small $\eps > 0$, and any correct $(\eps, \delta)$-estimator $A$, we can find a value $\bar{\eps} > \eps$ such that $A$ needs at least $\Omega(1/\bar{\eps}^2 \cdot \log\log\eps^{-1})$ samples to terminate on $\bar{D}(1/2 + \bar{\eps})$. On the other hand, for every distribution $\bar{D}(1/2 + \bar{\eps})$, there exists a specific $(\eps, \delta)$-estimator that terminates in just $O(1/\bar{\eps}^2)$ steps on this distribution.

\paragraph{Equivalence between the I/O and sequential estimation settings.} 

As we discussed in \Cref{sec:results_io}, instance-optimal algorithms for sequential estimation are intimately related to order-oblivious instance-optimal algorithms for I/O efficient estimation. Namely, the connection is roughly in that an instance-optimal algorithm in the I/O settings is equivalent to having an instance-optimal algorithm in the sequential estimation setting for a mixture version of the problem. Recall that with the mixture version of the problem, each sampled item from the input distribution $D$ is itself a distribution, and we want to solve the problem (e.g.\ approximate the median) on the mixture distribution $M(D)$.
On an intuitive level, sampling a distribution from $D$ in the sequential estimation setting then corresponds to sampling a block in the I/O setting.

We prove a general equivalence theorem whose informal version goes as follows. 

\begin{theorem}[Informal version of \Cref{thm:io_transfer}]
\label{thm:io_transfer_informal}
    The following holds under mild technical conditions.  Whenever we have an instance-optimal algorithm for a problem on mixture distributions in the setting of sequential estimation, we can also get an order-oblivious instance-optimal algorithm for the same problem in the model of I/O efficient algorithms, and vice versa.  
\end{theorem}
In order to prove this claim, we show two reductions: one that turns an algorithm in the I/O setting into one in the sequential estimation setting, and a second reduction in the other direction. The reductions work as follows.

Given an algorithm $A_{I/O}$ in the I/O model, we define a new algorithm $A_{dist}$ in the sequential estimation setting by simulating $A_{I/O}$: we (implicitly) construct a long input consisting of many samples from the input distribution; we then execute $A_{I/O}$ on this input. (Some technical issues arise here, leading to the theorem having mild technical assumptions.) 
On the other hand, given an algorithm $A_{dist}$ in the sequential estimation setting, we construct $A_{I/O}$ by randomly sampling blocks and feeding them into $A_{dist}$.

While this would already suffice to prove equivalence between the two settings in the worst-case sense, it is not enough in the context of instance-optimality. What we need to show is that when we use one of the above reductions, we reduce to an instance that has exactly (up to a constant) the same complexity -- otherwise we could reduce to some instance that is much harder than the instance we wanted to solve, in which case the resulting algorithm would not be instance-optimal.

If, hypothetically,  the reductions were inverses of each other, this would suffice. Namely, if we could reduce instance $x_{I/O}$ to $x_{dist}$ but also $x_{dist}$ to $x_{I/O}$, then any algorithm that solves one of the instances may solve the other in the same complexity. This means that the optimum for these two instances is the same. If this is the case for all instances in the two settings, then any instance-optimal algorithm in one setting is also instance-optimal in the other.

Sadly, the reductions are not inverses of each other, as we need. That is, if we take an input in one of the settings, we reduce it to an instance in the other setting and then use the other reduction to get an instance in the first setting, we will end up with a different instance than the one we started with. We show that, nevertheless, the complexity of the two instances is the same, and this suffices to prove the equivalence.

\subsection{Defining and motivating order-oblivious instance-optimality} \label{sec:model_motivation}


This section is devoted to the discussion of order-oblivious instance optimality in the context of I/O-efficient sampling. In particular, we argue why the guarantee of order-oblivious instance optimality is the ``right'' guarantee to strive for. 
We first give an informal definition of (order-oblivious) instance optimality. 

\begin{definition}[Informal definition of Instance-optimality and Order-oblivious instance optimality] \label{def:optimality_informal}
We say an algorithm is correct for $\eps,\delta$ if its error is at most $\eps$ with probability $\geq 1-\delta$ on any input.

We say that a correct algorithm $A$ is \emph{instance-optimal} (\Cref{def:instance_optimality}) if there is no possible input $I$ such that some correct algorithm would be faster on $I$ than $A$ (by an arbitrarily large constant).

An algorithm is \emph{order-oblivious instance optimal} (\Cref{def:oo-instance-optimality}) if the above holds if we measure the complexity of $A$ on a specific input as the worst case over all permutations of that input.
\end{definition}

%
%
%
%

\subsubsection{Motivation behind the definition}
%
Our main desideratum is to utilize to the fullest extent possible each block that we read. 
Unfortunately, if we use worst-case complexity as our complexity measure, we cannot get any benefit from reading the whole block instead of just one bit.
\begin{claim}
\label{cl:nvm}
Suppose we have a $0/1$ string of length $n$ divided into blocks, each containing $B$ bits. Suppose we want to estimate the fraction of $1$'s with success probability $2/3$. Then $\Theta(1/\epsilon^2)$ I/O's are both necessary and sufficient in the worst case if $B \leq o(\epsilon^2 n)$.
\end{claim}
\begin{proof}[Proof sketch]
Consider two possible inputs. In both of them, each block either contains only zeros or only ones. In the first input, the fraction of zero blocks is $1/2 - \eps$. In the second input, this fraction is $1/2 + \eps$. An algorithm for estimating a mean can be turned into an algorithm that distinguishes the two distributions over the input blocks. 
Folklore lower bound \cite{canetti1995lower} says that distinguishing the two inputs requires $\Omega(1/\eps^2)$ samples. 
\end{proof}

This means that we aim for a more fine-grained complexity measure than the worst-case complexity. 
Sadly, the notion of instance-optimality is too strong and it is impossible to achieve, as we now sketch. 

\begin{claim}
\label{cl:mvn}
Suppose we have a $0/1$ string of length $n$ and let $B = \omega(1)$ be the block size. Then there is no algorithm that would be instance-optimal up to constant factors for the problem of estimating the fraction of $1$'s.
\end{claim}
\begin{proof}[Proof sketch.]
%
Given a possible (``guessed'') input $x'$, we consider an algorithm that $\Theta(1/\eps)$ times samples an index $i$ and verifies that the $i$-th block of the true input $x$ and the guessed input $x'$ are equal. 
One can prove that if all sampled blocks were in fact equal for $x$ and $x'$, then with constant probability the inputs differ on at most $\eps$ fraction of blocks. This allows us to return the (pre-computed) correct answer for $x'$. 
If the verification fails, we run an arbitrary correct algorithm.

This way, we can, for any possible input $x'$, get a correct algorithm that runs on $x'$ in complexity $O(1/\eps)$. Any instance-optimal algorithm thus has to have a worst-case complexity of $O(1/\eps)$. This is in contradiction with the lower bound of $\Omega(1/\eps^2)$ from \Cref{cl:nvm}.
\end{proof}

Since we cannot get standard instance optimality, we need to relax our requirements somewhat. 
We see above that it is the access to the blocks in a fixed order that causes us problems: 
This leads us to the notion of order-oblivious instance optimality \cite{afshani2017instance} where, intuitively speaking, we view the input as a multiset and ignore the order.

\subsection{Related work} \label{sec:related_work}

We next discuss various related work. 

\subsubsection{I/O-efficient algorithms, distribution testing, and instance-optimality}

\paragraph{I/O-efficient algorithms.}
In computations, reading and writing to memory (I/O) is often the bottleneck. I/O-efficient, also called cache-aware, algorithms are designed to minimize the number of I/O's performed and to make efficient use of cache memory in computer systems.

Since the introduction of this problem to algorithm design by \citet{lam1991cache} and the formalization as a model of computation by \citet{frigo1999cacheoblivious}, this issue has been studied thoroughly. However, the current literature focuses on either $(a)$ data structures (such as $B$-trees, funnel heaps, and others, see \cite{arge2005cacheoblivious} for a detailed discussion), or $(b)$ on algorithms that run in at least linear time (such as selection, sorting, or matrix multiplication, see \cite{demaine2002cacheoblivious} for a detailed discussion).
We cannot hope to cover the whole area of I/O-efficient algorithms; the interested reader should look at the surveys \cite{arge2005cacheoblivious,demaine2002cacheoblivious} and the more up-to-date references in \cite{frigo2012cache}.

As we show in this paper, the fact of being able to read a block at a unit cost is also useful in the world of sublinear-time algorithms. This has been considered by the pioneering paper \cite{andoni2009external} which gives worst-case optimal algorithms for distribution testing on empirical distributions, namely for testing of distinctness, uniformity, and identity. The problem of sampling and I/O-efficient sublinear algorithms was also considered by \cite{ning2020range}, who, however, rely on an assumption that the data comes from a specific distribution. We are not aware of other algorithms relying on block-based memory access that have sublinear I/O complexity.

\paragraph{Learning theory, distribution testing, and instance-optimality.}
While learning theory has mostly focused on worst-case optimality in the past, there is also significant literature on instance-optimality.
Namely, instance-optimality has been studied for multi-armed bandits, with the first work in this direction going to 1985 \cite{lai1985asymptotically}.
More recently, the gap-entropy conjecture \cite{chen2016open} would imply the existence of an instance-optimal algorithm for identifying the best arm. Significant progress to proving this conjecture has been made by \citet{chen2017towards}. This problem has also been considered for some classes of multi-armed bandits, such as contextual bandits \cite{li2022instance}, linear bandits \cite{kirschner2021asymptotically}, or Lipschitz bandits \cite{magureanu2014lipschitz}. An algorithm that is almost instance-optimal is also known for reinforcement learning in Markov decision processes \cite{yin2021towards}. A result that generalizes some of the multi-armed bandit results, as well as reinforcement learning, has recently been given \cite{dong2022asymptotic}; the related work section of this paper is also a good place where an interested reader can learn more about the related literature.

\citet{valiant2016instance} give an instance-optimal algorithm for learning a distribution with respect to the $\ell_1$ norm \cite{valiant2016instance}. Compare this to our result (in the more general I/O model) for learning a distribution w.r.t.\ the $\ell_\infty$ norm from \Cref{cor:counting_ub}. An instance-optimal algorithm for compressed sensing signals drawn from a known prior distribution has been shown by \citet{jalal2021instance}.

Testing whether a discrete distribution is equal to a fixed distribution is a well-studied problem in distribution testing. In their seminal paper, \citet{valiant2017automatic} gave an instance-optimal algorithm for this problem. This is in contrast with the previously known worst-case optimal algorithm \cite{diakonikolas2018sample}.

\paragraph{Two papers closely related to ours}
After the acceptance of our manuscript, we learned about two additional closely related works from the area of instance-optimal sequential estimation.

The first paper from \citet{dagum2000optimal} considers the problem of estimating the mean with multiplicative error for which they provide an instance optimal algorithm (although they do not use the language of instance optimality). In contrast, our
paper focuses on additive error guarantees which allows our basic algorithms (see the algorithm from page \pageref{alg:mean_informal}) to be significantly simpler. Later
we extend our additive error
algorithm to the multiplicative
case, yeilding the same multiplicative guarantee as \citet{dagum2000optimal}.
We also extend our basic 
additive algorithm to the estimation of medians, quantiles, and learning distributions in $\ell_\infty$ and KS distance.

The second closely paper from \citet{dang2023optimalitymeanestimationworstcase} considered the same basic problem of instance-optimal algorithms for the mean estimation. 
However, their setting is slightly different: While in our paper we measure estimators by the number of samples they need given a fixed desired error $\eps$, \citet{dang2023optimalitymeanestimationworstcase} consider the setting where we want to achieve the best possible error rates given a fixed number of samples $n$. Unfortunately, they show that in this setting, instance optimality is not achievable. Their main result, \cite[Theorem 2]{dang2023optimalitymeanestimationworstcase}, is a version of our lower bound for mean estimation from \cref{thm:counting_lb}.  

Both of these papers also prove a lemma (\Cref{thm:counting_lb}) which we use to prove the optimality of our algorithms. The first paper proves it in the case of discrete distributions, while the second also extends it to general distributions. We give a proof in our write-up for the sake of completeness and because the more general statement in \cite{dang2023optimalitymeanestimationworstcase}

\subsubsection{Instance-optimality and sequential analysis in statistics}

\paragraph{Instance-optimality in statistics.} While the TCS community typically considers optimality in the worst case, other notions of optimality are used in statistics. One of the commonly used notions of optimality exactly corresponds to instance optimality. Most notably, in the parametric setting\footnote{That is, the distribution comes from a class parameterized by some parameter $\theta$; we want to estimate $\theta$ or its function, based on a sample}, people often consider the minimum-variance unbiased estimators (MVUE) \cite[7.3.2]{casella2002statistical}. This can be seen as instance optimality if an algorithm is said to be correct if it is unbiased. Notably, any efficient estimator (an unbiased estimator matching the famous Cramér-Rao lower bound) is MVUE. 

One of the major areas of classical statistics is the area theory of maximum likelihood estimation (MLE) \cite[7.2.2]{casella2002statistical}. It is known that, under very general technical conditions, a MLE estimator is asymptotically efficient, which can be seen as a non-uniform version instance optimality\footnote{Non-uniformity here means that the speed of convergence of the $o(1)$ can depend on the precise distribution.} up to a factor of $1+o(1)$ \cite{daniels1961asymptotic}.

\paragraph{Difference in guarantees in statistics and TCS.} Unfortunately, for example, the fact MLE's in general only achieve a \emph{non-uniform} version of instance optimality means that this result does not imply algorithms with meaningful guarantees: the non-uniformity corresponds to allowing the constant in the $\mathcal{O}$ notation to depend on the input size, allowing us to solve any computable problem in (non-uniform) constant time. Sadly, results of this type are very common in statistics. This is one of the reasons why we need to use techniques from computer science to get the desired results.

\paragraph{Sequential estimation.} 
All the above work in statistics assumes that the number of available samples is fixed, and that it cannot be picked adaptively. In this paper, we focus on the setting where we fix the desired accuracy and try to minimize the sample complexity. This is in statistics called sequential estimation, which we now discuss.
%
Our account is by no means comprehensive, as such treatment is more suitable for a book on the topic. The interested reader may wish to see \cite{lai2001sequential} for a survey of some of the known results.

Our work can be seen as revisiting the setting of sequential analysis in the non-asymptotic setting while using modern tools from the concentration of measure and distribution testing. The non-asymptotic sequential setting, not widely considered in the statistics literature, is very important from the perspective of computer science. Namely, the asymptotic results are always for a fixed distribution, and the error parameter going to zero. This corresponds to considering asymptotic under a fixed input. But we want our algorithms to work also for non-constant input sizes! For this reason, the non-asymptotic viewpoint is crucial for us. We contrast our needs with the results in the statistical literature in greater detail below, explaining why the setting assumed by the statisticians is not suitable for our use case.

The problem of sequential analysis is said to stretch at least back to the 17th century when it was considered in the context of the gambler's ruin \cite{ghosh1991brief}. Since then, it found many practical use cases. For example, it was used for quality control for anti-aircraft gunnery in the 2nd World War. We highlight one of the early results in the area: an optimal test that distinguishes between two distributions \citet{wald1945sequential}.  See the text \cite{lai2001sequential} for more detailed treatment. We now focus on the work in the sequential analysis setting on estimation -- usually called sequential estimation.

In his seminal paper, \citet{haldane1945method} gave an estimator that estimates the probability $p$ of a Bernoulli trial in sample complexity $O(p/\epsilon^2)$. More closely related to our work, under the assumption of the distribution being normal, an estimator for the mean has been given by \citet{stein1945two} that gives optimal confidence intervals.
This was later generalized by \citet{cox1952estimation}, who relaxed the parametric assumptions. However, his result still does not imply our result, as it assumes that the variance only depends on the parameter that we are estimating, which is not the case in this paper. Moreover, the sense in which this result is optimal is significantly weaker than our results: it shows optimality up to a constant, where the constant may depend on the given estimation task (but not on the accuracy parameter $\eps$).

Curiously, these two estimators are very similar to the one we use for mean estimation, even though it is a completely different setting (and a significantly different analysis). 

Specifically, the estimator has two phases where the number of samples in the second is determined by the first. Two decades later, this problem was solved in the asymptotic non-parametric setting (instead of non-asymptotic parametric from \cite{stein1945two}) by \citet{chow1965optimal}.

\section{Preliminaries}
\label{sec:preliminaries}

This section collects the basic notation and tools that we use, together with formal definitions of sequential estimation and the I/O estimation model and some general results concerning these models. We encourage the reader to read this section only briefly and return back to it when necessary -- the formal definitions are needed for formally stating our equivalence theorem, but may not be needed for casual reading. 

\subsection{Basic notation}
\jakub{Add notation for quantiles}
We use $[n] = \{1, 2, \dots, n\}$. For a finite set $S$ of reals, we use $\inf(S), \sup(S)$ to denote its infimum and supremum. We use $\hat{\mu}$ for final estimators, in two-phase algorithms, we use $\tilde{\mu}$ for the intermediate estimators. We use the notation $\hmu = \mu \pm \eps$ to denote the fact $\mu - \eps \le \hmu \le \mu + \eps$. Given a distribution $D$ over a set $S$ and a function $f: S \rightarrow S'$, we define a distribution $f \circ D$ to be the distribution we get by sampling from $D$ and applying $f$. Given a set $S$, we use $Dist(S)$ to denote the set of all distributions on $S$. Finally, given a real-valued random variable $X$ and any $k \ge 1$, we use $\|X\|_k$ to denote the $k$th moment of $X$, i.e., $\|X\|_k := \E[|X|^k]^{1/k}$. For two distributions $D_1,D_2$, we use $D_1 \otimes D_2$ to denote the product distribution of $D_1$ and $D_2$. We sometimes avoid floor/ceiling functions to avoid clutter in notation. We define a mixture distribution as follows:
\begin{definition}[Mixture distribution] \label{def:mixture}
Consider any distribution $D$ over $Dist(S)$, i.e., every sample from $D$ is itself a distribution over $S$. We define $M(D)$ to be the \emph{mixture} defined by $D$, defined by $P_{X \sim M(D)}[X \in S'] = \E_{Y \sim D}[P_{X \sim Y}[X \in S']]$ for all measurable subsets $S' \subseteq S$.
\end{definition}

\subsection{Tools}
\label{sec:lb_definitions}

We will use the following approximation of square root. 
\begin{fact}
\label{fact:sqrt}
For any $-1/2 \le \eps \le 1/2$, we have
\begin{align*}
    1 + \eps/2 - \eps^2/2 \le \sqrt{1 + \eps} \le 1 + \eps/2
\end{align*}
\end{fact}

We will need the standard Chebyshev's inequality.
\begin{lemma}[Chebyshev's inequality]
\label{lem:chebyshev}
For any random variable $X$ with mean $\mu$ and standard deviation $\sigma$, we have
\begin{align*}
    \P( |X - \mu| \ge t \sigma) \le 1/t^2
\end{align*}
\end{lemma}

We will also use the following formula for the variance of sample variance. 
\begin{lemma}[\cite{stackexchange_variance}]
\label{lem:variance_of_variance}
Let $X_1, \dots, X_n$ be i.i.d. samples from a distribution with mean $\mu$, variance $\sigma^2$ and fourth central moment $\mu_4 = \E[(X-\mu)^4]$. Let $\hat\mu = \frac{1}{n} \sum_{i=1}^n X_i$ be the sample mean and  $\hat\sigma^2 = \frac{1}{n-1} \sum_{i = 1}^n (X_i - \hat\mu)^2$ be the sample variance. 
Then both $\hat\mu$ and $\hat\sigma^2$ are unbiased estimators of $\mu$ and $\sigma^2$, respectively. Moreover, 
\[
\Var[\hat\sigma^2] = \frac{1}{n} \cdot \mu_4 - \frac{n-3}{n(n-1)} \cdot \sigma^4. 
\]
\end{lemma}

We will also need standard Chernoff bounds and Bernstein's inequality. 

\begin{lemma}[Chernoff bound]
\label{lem:chernoff}
For any $X = X_1 +  \dots + X_T$ where each $X_i$ is i.i.d. Bernoulli random variable with $\E[X_i] = \mu$, we have
$$
\P( (X - T\mu) > \Delta) \le 2 \exp\left( - 2 \Delta^2 / n\right)
$$
\end{lemma}

\begin{lemma}{\cite{dubhashi2009concentration}}
    \label{lem:bernstein}
    For any $X = X_1 + \dots + X_T$ where $X_i$ are i.i.d. with $\E[X_i] = \mu$, $\Var[X_i] = \sigma^2$, and with $0 \le X_i \le 1$ almost surely. Then for any $ \Delta > 0$ we have
    \begin{align*}
        \P( |X - T\mu| > \Delta) = \exp\left(- \Omega\left( \min\left( \Delta, \frac{\Delta^2}{T \cdot \sigma^2} \right)\right) \right)
    \end{align*}
\end{lemma}

We prove the following result about the maximum of random variables in \Cref{sec:deferred_proofs}.
\begin{restatable}{lemma}{maxlemma} \label{lem:expectation_of_max}
Let $X_1, \dots, X_k$ be independent random variables. It holds that $\E[\max_i X_i] = O(\max_{i} \|X_i\|_{\log k})$.
\end{restatable}

\paragraph{Mesuring distribution distance}

We will follow the lower bound methodology based on Hellinger distance that was developed in \cite{bar2001sampling}. 

Let $p$ and $q$ be two discrete distributions. In the following definitions we always assume that relevant distributions are over the same support $[t]$, for some integer $t$. Let $p_i := \P(p = i)$ and $q_i = \P(q = i)$. The \emph{squared Hellinger distance} between $p$ and $q$ is defined as
\begin{align}
\label{def:hellinger}
    H^2(p, q) = 1 - \sum_{i=1}^{t} \sqrt{p_i \cdot q_i} = \frac{1}{2} \cdot \sum_{i=1}^t (\sqrt{p_i}-\sqrt{q_i})^2. 
\end{align}

We will need two important facts about the squared Hellinger distance. The first fact nicely bounds the squared Hellinger distance of product distributions. The second fact relates the squared Hellinger distance to the total variation distance.

\begin{fact}[Equation (C.4) in~\cite{canonne2020survey}]
    \label{fact:hellinger_product}
    For two distributions $p, q$ over $[t]$ and $m \ge 1$ any integer, define $p^{\otimes m}$ to be the distribution over $[t]^m$ defined by drawing $m$ independent samples $s_1, \dots, s_m$ from $p$ and returning $(s_1, \dots, s_m) \in [t]^m$ (and define $q^{\otimes m}$ similarly).
    Then, we have 
    \begin{align*}
        H^2\left(p^{\otimes m}, q^{\otimes m}\right) = 1 - \left(1 - H^2(p, q)\right)^m.
    \end{align*}
\end{fact}

\begin{fact}[Theorem C.5 in~\cite{canonne2020survey}]
    \label{fact:hellinger_implies_l1}
    For any two distributions $p,q$ over some discrete set $[t]$, 
    \begin{align*}
        1 - d_{TV}(p, q) \le 1 - H^2(p, q) \le \sqrt{1 - d_{TV}(p, q)^2}.
    \end{align*} 
\end{fact}

Finally, we use bounds on total variation distance to rule out \emph{distinguishers} between two distributions. Formally, a $1-\delta$ distinguisher between $p$ and $q$ is a function assigning a label \texttt{p} or \texttt{q} to every element in the support of $p$ (or $q$). We require that if we sample from $p$, the distinguisher outputs \texttt{p} with probability at least $1-\delta$ and the same for $q$. The following is a standard fact about total variation distance. 

\begin{fact}
    \label{fact:l1}
    Let $p$ and $q$ be any two distributions with $d_{TV}(p, q) \le 1 - \delta$. Then, there is no $1-\delta$ distinguisher between $p$ and $q$. 
\end{fact}

\subsection{Sequential estimation setting}
\label{sec:sampling_definitions}

We now formally define the model for estimating a parameter of a distribution by sampling and the notion of instance optimality in this setting. As we said above, this formalism is needed to state our equivalence theorem, but may not be necessary for casual reading, meaning that the reader may wish to skip this section.

First, we define a parameter of a distribution. The set $K$ below is typically the set of real numbers. 
\begin{definition}[A distribution parameter]
    Given sets $S$ and $K$, we say that a function $\theta : Dist(S) \rightarrow K$ is a distribution parameter. 
\end{definition}

We next define a distance function that measures how far an estimate of a parameter is from the truth.  

\begin{definition}[Distance function]
\label{def:quasimetric}
    Given a distribution parameter $\theta : Dist(S) \rightarrow K$, a \emph{distance function for $\theta$} is a family $d$ of functions $d_D: K \rightarrow \R$ for $D \in Dist(S)$. Each $d_D$ needs to satisfy that for every $\hat\theta \in K$ we have $d_D(\hat\theta) \ge 0$ and that $d_D(\theta(D)) = 0$.  
\end{definition}

We can now formally define a \emph{sequential estimation problem}. A \emph{sampling algorithm} is any process that has sampling access to an input distribution $D$; the algorithm samples elements from $D$ until it decides to finish and return a value. 

\begin{definition}[Distribution estimation problem]
\label{def:dist_problem}
    A sequential estimation problem $P$ is formally a triple $P = (S, \theta, d )$ where $\theta: Dist(S) \rightarrow K$ is a distribution parameter and $d$ is a distance function for $\theta$.
    
    The task is to estimate $\theta$ for any distribution in $Dist(S)$ using as few samples from an input distribution $D$ as possible. 
    Formally, we say that an algorithm is $(\eps,\delta)$-correct for $P$ if for any distribution $D \in Dist(S)$, it holds for the output $\hat\theta$ of the algorithm that $P[d_D(\hat\theta) > \eps] < \delta$.
\end{definition}

Let us now list the problems that we consider in this paper. 
\begin{enumerate}
    \item \mean: This problem is defined relative to a set $S \subseteq \R$ being any subset of reals. The function $\theta$ maps a distribution $D$ over $S$ to its mean $\mu(D)$ and $d_D$ maps an estimate $\hat\mu$ to $|\hat\mu - \mu(D)|$.
    \item \freq: This is a special version of the \mean problem with $S = \{0, 1\}$.\footnote{Note that an instance optimal algorithm $A$ for \mean over $S$ is not automatically instance optimal for \mean over sets $S' \subseteq S$; for $S'$, the set of correct algorithms that $A$ needs to be competitive with is potentially larger. Issues like this are the reason why we try to keep this section formal. }
    \item \distinfty: $S$ is an arbitrary finite set, $\theta$ is identity and $d_D$ maps $D'$ to $\ell_\infty(D, D')$. 
    \item \quantile: This problem is defined relative to a number $q \in [0,1]$. We have $S = \R$, $\theta$ maps $D$ to its $q$-th quantile $Q_D(q)$ and $d_D$ maps an estimate $\hat N$ to  $\max(q-P_{X \sim D}(X \leq \hat N), P_{X \sim D}(X < N) - q)$.  
    \item \distks: We have $S = \R$, $\theta$ is identity and $d_D$ maps $D'$ to the Kolmogorov-Smirnov distance between $D$ and $D'$, i.e., $d_D(D') = \sup_{t \in \R} | P_{X \sim D}(x \le t) - P_{X \sim D'}(x \le t)|$. 
        
 \end{enumerate}



Let $Dist_B(S)$ be the set of distributions on $D$ where all probabilities are multiples of $1/B$.
We also define a $B$-\textsc{Mixture} version of each of the above problems. These problems are analogous to \distinfty, \quantile, and \distks, but the input distribution over a set $S$ is replaced by an input distribution $D$ over a set $\fS = Dist_B(S)$ and we want to solve the original problem for the mixture distribution $M(D)$. We leave out the ``$B$-'' to denote the limiting case where the probabilities may be arbitrary. 

With a slight abuse of notation, we denote the set of the $B$-\textsc{Mixture} problems for all $B \in \mathbb{N}$ as the \textsc{Step-Mixture} version of the problem. Indeed, this is technically not a problem in the above sense, but rather a set of problems. This will, however, not be an important distinction for us.


Finally, we also define variants of the problems with multiplicative error where we divide the value of $d_D(x)$ by $\theta(D)$.

\paragraph{Instance optimality}
Next, we define formally what it means that an algorithm for an estimation problem is instance optimal. We use $T(A_{\eps, \delta}, D)$ to denote the expected number of samples that a sampling algorithm $A$ makes on an input distribution $D$ for input precision parameters $\eps, \delta > 0$. For a sequential estimation problem $P$ and an input distribution $D$ we define: 
\begin{align} \label{def:opt}
    OPT_{\eps, \delta}(P, D) = \inf_{A \text{ is $(\eps,\delta)$-correct for $P$}} T(A_{\eps, \delta}, D).
\end{align}

\begin{definition}[Instance optimality of estimators]
\label{def:instance_optimality}
For any sequential estimation problem $P$ from \Cref{def:dist_problem} we say that an algorithm $A_{\eps, \delta}$ is an instance optimal $(\eps,\delta)$-estimator for $P$ if:
\begin{enumerate}
    \item Correctness: For any $\eps, \delta > 0$, we have that $A_{\eps, \delta}$ is an $(\eps, \delta)$-correct for $P$. 
    \item Competivity: There is a constant $C$ such that for every distribution $D$ and every $\eps, \delta > 0$, we have $T(A_{\eps, \delta}, D) \le C \cdot OPT_{\eps, \delta}(P, D)$.  
\end{enumerate}
We say that $A$ is instance optimal for $P$ if for $\eps, \delta > 0$, the algorithm $A_{\eps, \delta}$ parametrized by $\eps,\delta$ is instance optimal $(\eps,\delta)$-estimator. 


\end{definition}

If we replace the constant $C$ by a function $C(\eps, \delta)$, we get algorithms that are \emph{instance optimal up to} $C(\eps, \delta)$.

\paragraph{Enhanced probability amplification}

Next, we state a useful theorem that enables us to amplify the success probability of an algorithm while also amplifying the concentration of the running time of the algorithm. Our theorem is based on the standard median amplification trick but it in addition bounds the tails of the complexity which comes useful in some applications.  

%

In order for the amplification theorem to work, we need a mild technical assumption about the problem. We say that a sequential estimation problem $P = (S, \theta, d_D)$ is \emph{monotonic} if the image of $\theta$ is $\R$ and if $d_D(\hat\theta)$ is non-decreasing in $|\hat\theta - \theta(D)|$. 
In the next lemma, we use the notation $\mathbf{T}(A, D)$ for the random variable representing the sample complexity of $A$ on $D$ (whose mean is $T(A,D)$). 

\begin{restatable}{lemma}{amplification}\label{lem:prob_amplification}
Suppose that an algorithm $A$ is $(\eps, 1/10)$-correct for a monotonic sequential estimation problem $P$. 
Then for any $\delta > 0$, there is an $(\eps, \delta)$-correct algorithm $A'$ for $P$. 
Moreover, we have for any distribution $D$ that
\begin{align*}
T(A', D) &\le \|\mathbf{T}(A', D)\|_{O(\log\delta^{-1})}
= O\left(T(A,D) \cdot  \log \delta^{-1} \right)\,.\end{align*}
\end{restatable}
The proof is deferred to \Cref{sec:deferred_proofs}.

\subsection{I/O estimation model}
\label{sec:io_definitions}

We next formally define the I/O estimation model, complexity, and estimation problems. Again, this formalism is needed to state our equivalence theorem, but may not be necessary for casual reading.

The definitions are analogous to definitions from \Cref{sec:sampling_definitions}. In the following definition,     $\Pi_n$ stands for the set of all permutations on $\{1, \cdots, n\}$ elements and a concrete permutation $\pi \in \Pi_k$, we use $\pi(I)$ for an input consisting of $k$ blocks, to denote the input that we get by shuffling the blocks of an input $I$ by $\pi$. We also use the notation $\fD(I)$ for a sequence $I$ of length $n$ that denotes the uniform distribution over the elements of $I$.  

We start by defining the I/O model of computation. Essentially, the definition says that we use Word-RAM but input reads are performed in blocks of $B$ words.
\begin{definition}[I/O model]
\label{def:io_model}
    The \emph{I/O model} is an extension of the standard Word-RAM model parametrized by an additional parameter $B$, the length of a block. 
    The input of length $n$ consists of $\lceil n/B \rceil$ blocks where each block contains a sequence of $B$ words. 
    The algorithm can access the input by a READ($i$) operation that takes an integer $i$ between $1$ and $\lceil n / B \rceil$ and loads the $B$ words at position $i$ of the input to the working memory. 
\end{definition}

\begin{definition}[I/O complexity] \label{def:io_complexity}
    The \emph{I/O complexity} $T'(A, I)$ of an algorithm $A$ on input $I$ is the expected number of its READ operations on $I$. 
    The \emph{order-oblivious I/O complexity} is defined as $T(A, I) = \max_{\pi \in \Pi_{\lceil n/B \rceil}} T'(A, \pi(I))$. 
\end{definition}

\begin{definition}[I/O estimation problem] \label{def:io_problem}
    An \emph{I/O estimation problem} $P$ is formally a quadruple $P = (B, S, \theta, d)$ where $B$ is an integer, $S$ an arbitrary set, $\theta$ a distribution parameter, and $d$ is a distance function for $\theta$.
    
    The input to $P$ is an arbitrary sequence $I$ of blocks of $B$ values from $S$, with at most one block being allowed to have fewer than $B$ elements.
The goal is to construct an $(\eps, \delta)$-correct estimator defined as an estimator such that for its output $\hat\theta$ we have $P[ d_{\fD(I)}(\hat\theta) > \eps] < \delta$. 
\end{definition}

Note that any I/O-estimation is permutation-invariant in the sense that the distance function does not depend on the permutation.

From now on, we will work only with the order-oblivious I/O complexity and, in particular, define the optimal complexity $OPT_{\eps, \delta}$ with respect to the order-oblivious complexity exactly as in the sequential estimation setting (\Cref{def:opt}). Using this, we now define order-oblivious instance optimality. The fact that we use the order-oblivious complexity is crucial to make our equivalence theorem (below) work and to ensure that such instance-optimal algorithms even exist.
\begin{definition} \label{def:oo-instance-optimality}
An algorithm $A_{\eps, \delta}$ in the I/O model is \emph{order-oblivious instance optimal} $(\eps,\delta)$-estimator if
\begin{enumerate}
    \item Correctness: For any $\eps, \delta > 0$, we have that $A_{\eps, \delta}$ is an $(\eps, \delta)$-correct on $P$. 
    \item Competitivity: There is a constant $C$ such that for every input $I$ and every $\eps, \delta > 0$, we have $T(A_{\eps, \delta}, I) \le C \cdot OPT_{\eps, \delta}(P, I)$.  
\end{enumerate}
\end{definition}

For each problem in the sequential estimation setting, we define its \textsc{Block} variant in the I/O estimation setting. We define this as the distributional problem on a distribution obtained by sampling a uniform block and then sampling a uniform element from that block.  Formally, if we have problem $P_{dist} = (S, \theta, d)$, we define a problem in the I/O setting for any $B \in \mathbb{N}$ as $(B, S, \theta, d')$ where $d'_{\mathcal{D}(I)} = d_D$ for $D$ a uniformly random element of a uniformly random block of $I$, where $I \in S^B$ is an input.

\paragraph{Equivalence between the I/O and sequential estimation settings}

We will next state the formal version of \Cref{thm:io_transfer_informal}. Before we do that, we need to specify mild technical conditions under which the equivalence between instance optimality in the two settings holds \footnote{Even milder conditions would suffice, but we are not aware of any reasonable problems where our requirements fail. }. These technical conditions are necessary because of some technicalities related to ``mismatches'' between our two models -- while the sequential estimation setting works with an arbitrary distribution, the I/O estimation model works only with an implicit uniform distribution over $n$ memory blocks. 

For the following definition, we use a notation as follows. For a distribution $D$, letting $X_1, \cdots, X_n \sim D$, we use $D_n$ to denote the uniform distribution on $\{X_1, \cdots, X_n\}$. We note that $D_n$ is itself a distribution over distributions. We naturally extend our notation, e.g., $d_{D_n}(\hat\theta)$ for fixed $\hat\theta$ is a random variable. We also use the notation $\fI(D_n)$ to denote the external memory input we get by arbitrarily listing the support of $D_n$ as an array of length $n$. For a set of values $x^i$, each associated with a sequence of random variables $X_1^i, X_2^i, \dots$, we say that \emph{$X^i_j \rightarrow x^i$ for all $i$ uniformly in probability} if for every $\eps$ there is $n$ such that for all $i$ we have $\P(|X^i_n - x^i| > \eps) < \eps$.

\begin{definition}[Consistency of a problem]
\label{def:consistent}
We say that a sequential estimation problem $P = (S, \theta, d_D)$ is consistent if for $n \rightarrow \infty$, it holds 
\begin{enumerate}
    \item For every  $\hat\theta$ we have $d_{D_n}(\hat\theta) \rightarrow d_{D}(\hat\theta)$ uniformly in probability.  \label{item:2}
\end{enumerate}
Consistency of an I/O estimation problem is defined analogously.
\end{definition}

The second technical condition ensures that the sample complexity does not have ``big jumps". \jakub{I don't like this name. I would assume it means something very different.}
\begin{definition}[Uniform boundedness of the sample complexity]
\label{def:bounded}
We say that the sample complexity of a sampling problem $P$ is uniformly bounded if for every $\eps, \delta > 0$ there is a constant $C_{\eps, \delta} > 0$ such that for all distributions $D$ we have $OPT_{\eps, \delta}(P,D) \le C_{\eps, \delta} \cdot OPT_{2\eps, 2\delta}(P,D)$. 
The uniform boundedness of an I/O estimation problem $P$ is defined analogously, replacing $OPT(P,D)$ with $OPT(P,I)$. 
\end{definition}

We can now state equivalence (up to some technical assumptions) between the sequential and I/O estimation models. This is the claim that justifies why we mostly focus on the sequential estimation setting. 

\begin{restatable}{theorem}{transtheorem}{\upshape (Equivalence between I/O model and sequential estimation)}{\bf.}
\label{thm:io_transfer}
Consider any pair of problems $P_{dist}, P_{I/O}$ where $P_{dist} := (Dist_B(S), \theta, d_D)$ is a sequential estimation problem and $P_{I/O} = (B, S, \theta, d_D)$ is an I/O estimation problem. \jakub{There is a mistake here. Clearly, we can't use $d_D$ in both, as we are using this on different sets. But I would skip fixing this for now maybe... It's very much a surface-level issue}
    \begin{enumerate}
    \item Assume that $P_{dist}$ is consistent per \Cref{def:consistent} and its sample complexity is bounded per \Cref{def:bounded}. Then, an instance optimal $(\eps, \delta)$-estimator (up to $c_{\eps,\delta}$) for $P_{dist}$ implies an order-oblivious instance optimal $(2\eps, 2\delta)$-estimator (up to $c_{\eps,\delta}$) for $P_{I/O}$.
        \item Assume that $P_{I/O}$ is consistent per \Cref{def:consistent} and its sample complexity is bounded per \Cref{def:bounded}. Then, an order-oblivious instance optimal $(\eps, \delta)$-estimator (up to $c_{\eps,\delta}$) for $P_{I/O}$ implies an instance optimal $(2\eps, 2\delta)$-estimator (up to $c_{\eps,\delta}$) for $P_{dist}$.    
    \end{enumerate}
\end{restatable}
We postpone the proof to \Cref{sec:deferred_proofs}.

\section{Mean Estimation and Counting} \label{sec:means}

In this section, we present an instance-optimal algorithm for the problem of estimating mean, i.e., the problem \mean, and the related problems. 
First, we present the algorithm for \mean in the sequential estimation setting in \Cref{sec:mean_upper}. The matching lower bound is proven in \Cref{sec:mean_lower}. \Cref{sec:mean_block} uses our general theorems to extend our algorithms to the block version of the problem and to the I/O estimation model. In \Cref{subsec:histogram}, we discuss the problem of learning a histogram and in \Cref{sec:mean_multiplicative} we show how our guarantee of small additive error can be replaced by the guarantee of a small multiplicative error. 


\subsection{Upper bound}
\label{sec:mean_upper}
In this subsection, we show an algorithm for mean estimation, i.e., the problem \mean, demonstrating an $O\left(\frac{1}{\eps} + \frac{\sigma^2}{\eps^2}\right)$ upper bound. 
The algorithm is described in \Cref{alg:counting_ones}, and works as follows. First, we get $O(1/\eps)$ samples from $D$ to get a preliminary estimate $\tilde{\mu}, \tilde{\sigma}^2$ of the mean and the variance of $D$. Then, using a fresh batch of $O\left(\frac{1}{\eps} + \frac{\tilde{\sigma}^2}{\eps^2}\right)$ samples, we compute a final estimate $\hat{\mu}$ using the sample mean.

\begin{algorithm}
\caption{Estimate $\mu(D)$ by $\hat{\mu}$ such that $\E[(\hat{\mu} - \mu(D))^2] \leq O(\epsilon^2)$} \label{alg:counting_ones}
$T_1 \leftarrow \frac{1}{\eps}$ \\
$X_1, \cdots, X_{T_1} \leftarrow$ get $T_1$ samples from $D$\\
Compute $\tilde{\mu} \leftarrow \frac{1}{T_1} \sum_{i=1}^{T_1} X_i$ \label{line:empirical_variance}\\
$\tilde{\sigma}^2 \leftarrow \frac{1}{T_1-1} \sum_{i=1}^{T_1} (X_i-\tilde{\mu})^2$\\
$T_2 \leftarrow \frac{1}{\eps} + \frac{\tilde{\sigma}^2}{\eps^2}$\\
$X_{T_1+1}, \cdots, X_{T_1+T_2} \leftarrow$ get $T_2$ new samples from $D$\\ 
\Return{$\hat{\mu} = \frac{1}{T_2} \sum_{i=T_1+1}^{T_2} X_i$} \\
\end{algorithm}


\begin{theorem} \label{thm:mean_estimation_ub}
Consider any distribution supported on $[0,1]$ with mean $\mu$ and variance $\sigma^2$. \Cref{alg:counting_ones} returns an unbiased estimate $\hat{\mu}$ of $\mu$ with standard deviation of $\sqrt{\E[(\hat{\mu}-\mu)^2]} \leq \sqrt{6} \cdot \epsilon$. Its expected sample complexity is $\frac{\sigma^2}{\epsilon^2} + \frac{2}{\epsilon}$.
\end{theorem}

In the following proof, we use $\mu, \sigma^2$ to represent the true mean and variance of the distribution $D$. We also use $\E_1, \E_2$ whenever we want to stress that the randomness is only over the first or the second phase of the algorithm. 

\begin{proof}

We start by bounding the expected query complexity.
The algorithm performs $\tilde\sigma^2/\eps^2 + 2/\eps$ samples. As $\tilde\sigma^2$ is an unbiased estimator of $\sigma^2$, we conclude that the expected sample complexity is $\sigma^2/\eps^2 + 2/\eps$, as desired. 

We next prove correctness. We first argue that the estimate $\hat\mu$ is unbiased. Note whatever values $x_1, \dots, x_{T_1}$ were sampled in the first phase, we have 
\[
\E[X_i | X_1 = x_1, \dots, X_{T_1} = x_{T_1}] = \mu
\]
for each $T_1 + 1 \le i \le T_1 + T_2$. Hence, $\E[\hat{\mu}|X_1 = x_1, \dots, X_{T_1} = x_{T_1}] = \mu$, so by the law of total expectation $\E[\hat\mu] = \mu$. 


We now bound the variance. As $\E[\hat{\mu}]=\mu$, we have
\[\Var[\hat{\mu}] = \E[(\hat{\mu}-\mu)^2] = \E_{1}[\E_{2}[(\hat{\mu}-\mu)^2 | T_2]],\]
where the last line follows from the law of total expectation.
Conditioned on $T_2 = t_2$%
, we have that $\E_2[(\hat{\mu}-\mu)^2|T_2 = t_2] = \frac{\sigma^2}{t_2}$, so we can write
\begin{equation} \label{eq:mean-ub-eq1}
\Var[\hat{\mu}] = \E_1\left[\frac{\sigma^2}{T_2}\right] = \E_1\left[\frac{\sigma^2}{\frac{1}{\eps} + \frac{\tilde{\sigma}^2}{\eps^2}}\right].    
\end{equation}

Next, we need to bound $\tilde{\sigma}^2$. Using the formula from \Cref{lem:variance_of_variance}, we have 
\begin{align*}\Var[\tilde{\sigma}^2] =& \frac{1}{T_1} \E_{X \sim D} [(X-\mu)^4]  - \frac{T_1-3}{T_1(T_1-1)} \cdot \E_{X \sim D} [(X-\mu)^2]^2 \leq \frac{1}{T_1} \E_{X \sim D} [(X-\mu)^4].\end{align*}
Since we always have $|X-\mu| \le 1$, we have $(X-\mu)^4 \le (X-\mu)^2,$ and we can further bound 
\[
\Var[\tilde{\sigma}^2]  \le \frac{1}{T_1} \cdot \E_{X \sim D} [(X-\mu)^2] = \eps \cdot \sigma^2.
\]

Therefore, by Chebyshev's inequality \Cref{lem:chebyshev} and using that $\tilde{\sigma}^2$ is unbiased by \Cref{lem:variance_of_variance}, it holds that 
\[\P\left(\tilde{\sigma}^2 < \frac{\sigma^2}{2}\right) \le \frac{\eps \cdot \sigma^2}{\sigma^4/4} = \frac{4 \eps}{\sigma^2}.\]

Let $\mathcal{E}_1$ be the event that $\tilde{\sigma}^2 < \frac{\sigma^2}{2}$. In this event, we have 
$$\frac{\sigma^2}{(1/\eps)+(\tilde{\sigma}^2/\eps^2)} \le \sigma^2 \cdot \eps.$$ Otherwise, 
$$\frac{\sigma^2}{(1/\eps)+(\tilde{\sigma}^2/\eps^2)} \le \frac{\sigma^2}{(\sigma^2/2)/\eps^2} \le 2 \eps^2. $$
Therefore,
\begin{align} \label{eq:mean-ub-eq2}
    \E\left[\frac{\sigma^2}{\frac{1}{\eps}+\frac{\tilde{\sigma}^2}{\eps^2}}\right] 
   \le& \P(\mathcal{E}_1) \cdot \sigma^2 \eps + (1-\P(\mathcal{E}_1)) \cdot 2 \eps^2 \nonumber \\ \leq& \frac{4 \eps}{\sigma^2} \cdot \sigma^2 \eps + 2 \eps^2 \le 6 \eps^2.
\end{align}

Combining Equations \eqref{eq:mean-ub-eq1} and \eqref{eq:mean-ub-eq2} finishes the proof. 
\end{proof}

By Chebyshev's inequality, we can give a straightforward conversion to a constant success probability algorithm. We can amplify its probability to $1-\delta$ by \Cref{lem:prob_amplification}. Finally, on an input distribution over any set $S$, we can first rescale the set to $S'$ with $\inf(S') = 0, \sup(S') = 1$ and then run our algorithm with $\eps' = \eps / (\sup(S) - \inf(S))$. We thus get the following corollary. 

\begin{corollary} \label{cor:counting_ub}
    There is an $(\eps,\delta)$-estimator for the problem \mean over any set $S$ in the sequential estimation setting such that for any distribution $D$, its sample complexity is
    \[
    O\left(\left(\frac{\sup(S) - \inf(S)}{\eps}+\frac{\sigma(D)^2}{\eps^2}\right) \log \delta^{-1} \right)\,.
    \] 
\end{corollary}

\subsection{Lower bound}
\label{sec:mean_lower}

Next, we prove a matching lower bound to \Cref{cor:counting_ub}. A version of the following result for discrete distributions was proven in \cite{dagum2000optimal} and this was later extended to general distributions in \cite{dang2023optimalitymeanestimationworstcase}. We include a proof for the sake of completeness.

\begin{proposition}
\label{thm:counting_lb}
For every $(\eps, \delta)$-estimator $A$ for the problem \mean over a set $S$ and any distribution $D$, we have
\begin{align*}
T(&A_{\eps, \delta}, D) = \Omega\left( \left( \frac{\sup(S) - \inf(S)}{\eps} + \frac{\sigma(D)^2}{\eps^2}\right) \cdot \log 1/\delta \right)    . 
\end{align*}
\end{proposition}
Specifically, we prove the following lemma which implies this proposition. We state this lemma separately as we will need this in \Cref{sec:quantiles}.
\begin{lemma} \label{lem:counting_lb_2}
For every distribution $D$ with support $S$, there exists a distribution $D'$ such that $(1)$ it holds $|\mu(D) - \mu(D')| \geq \eps$ and $(2)$ any algorithm that can distinguish the two distributions with probability $1-\delta$ has to use 
\[\Omega\left( \left( \frac{\sup(S) - \inf(S)}{\eps} + \frac{\sigma(D)^2}{\eps^2}\right) \cdot \log 1/\delta \right)\]
samples.
\end{lemma}
\begin{proof}
We first focus on distributions with finite support. We then show below that this can be used to prove the claim also for distributions with infinite support.
We use $\mu, \sigma$ for $\mu(D), \sigma(D)$. 
We assume without loss of generality that $\inf(S) = 0, \sup(S) = 1$; we can then get the general result by scaling. 
For any input distribution $D$ over $S$, we construct a new distribution $D'$ over $S$ satisfying
\begin{align*}
|\mu(D) - \mu(D')| = \Omega(\eps).  
\end{align*}
Thus, any algorithm that correctly solves the problem on $D$ for some $\eps_0 = \Omega(\eps)$ can be turned into a distinguisher (see \Cref{sec:lb_definitions} for definition) between $D$ and $D'$. We then show that these two distributions are hard to distinguish, thus proving the lower bound.


    \textbf{The first lower bound $\Omega(1/\eps \cdot \log 1/\delta)$: }
    The distribution $D'$ is constructed as follows: with probability $1 - \eps$, sample from $D$. In the remaining $\eps$ probability event, if $\mu(D) < \frac{1}{2}$ we return $1$, and if $\mu(D) \ge \frac{1}{2}$ we return $0$.
    On one hand, we have $\mu(D') < \mu(D) - \Omega(\eps)$, or $\mu(D') > \mu(D) + \Omega(\eps)$ in the alternative case. 
    On the other hand, any algorithm that samples at most $\frac{ \log_2 1/\delta}{10\eps}$ samples from $D'$ returns only samples from $D$ with probability
    \begin{align*}
        \left(1 - \eps \right)^{\frac{ \log_2 1/\delta}{10\eps}} \ge \delta. 
    \end{align*}
    Thus, it cannot distinguish between $D$ and $D'$ with $1-\delta$ success probability.  

    \textbf{The second lower bound $\Omega(\sigma^2/\eps^2 \cdot \log 1/\delta)$: }
    We assume $\sigma^2 > 2\eps$, otherwise the first lower bound dominates. 
    For any value $x \in S$, consider the deviation from the mean
    \begin{align*}
     d_x = x - \mu.   
    \end{align*}
    In the distribution $D'$ over $S$, we 
    replace the probability $p_x$ of sampling a value $x \in S$ as follows.
    \begin{align*}
        p_x^{D'} := p_x^D \cdot \left(1 + \frac{\eps}{\sigma^2}d_x \right).
    \end{align*}
    
    \textbf{$D'$ is well-defined:} To see that the construction of $D'$ is well defined, recall our assumption that $\sigma^2 > 2\eps$, so $\frac{\eps}{\sigma^2} < \frac{1}{2}$. 
    Moreover, $|d_x| = |x-\mu| < 1$, and we thus have $p_x^{D'} \ge p_x^D \cdot \left( 1 - \frac{1}{2} \cdot 1\right) \ge 0$. 
    Moreover, 
    using that 
    \begin{align}
        \label{eq:handy}
    \sum_{x \in S} p_x^D \cdot (x-\mu) = \sum_{x \in S} p_x^D x - \sum_{x \in S} p_x^D \mu = \mu - \mu = 0, 
    \end{align}
    we conclude that
\[\sum_{x \in S} p_x^{D'} = \sum_{x \in S} p_x^D + \frac{\eps}{\sigma^2} \cdot \sum_{x \in S} p_x^D \cdot (x-\mu) = \sum_{x \in S} p_x^D = 1,\]
i.e., the probabilities of $D'$ are indeed normalized to sum up to $1$. 

    \textbf{$\mu(D')$ is far from $\mu(D)$:} Next, we need to check that $\mu(D') > \mu(D) + \Omega(\eps)$. By construction of $D'$ we have:
    \begin{align*}
    \mu(D') - \mu(D) &= \sum_{x \in S} (p_x^{D'}-p_x^D) x 
     = \frac{\eps}{\sigma^2} \cdot \sum_{x \in S} p_x^D \cdot (x - \mu) \cdot x 
    \end{align*}
    Using \Cref{eq:handy} once more, we can rewrite this as 
    \begin{align*}
        \mu(D') - \mu(D) &=
         \frac{\eps}{\sigma^2} \cdot \sum_{x \in S} p_x^D \cdot (x - \mu) \cdot (x - \mu) 
         = \frac{\eps}{\sigma^2} \cdot \sigma^2 = \eps.
    \end{align*}

    \textbf{Hardness of distinguishing: }
    Finally, we argue that it is impossible to distinguish $D'$ from $D$ with $o(\sigma^2/\eps^2 \cdot \log 1/\delta)$ samples. To do so, we need to recall the Hellinger distance approach introduced in \Cref{sec:lb_definitions}.
    We can compute
\begin{align}
\label{eq:hellinger-bash-1}
    H^2(D, D') &= \frac{1}{2} \sum_{x \in S} \left(\sqrt{p_x^D}-\sqrt{p_x^{D'}}\right)^2 = \sum_{x \in S} p_x^D \cdot \left(\sqrt{1+\frac{\eps}{\sigma^2} d_x} - 1\right)^2.
\end{align}
    Because $\frac{\eps}{\sigma^2} < \frac{1}{2}$ and $|d_x| \le 1$, we have for $y := \frac{\eps}{\sigma^2} d_x$ that $|y| \le 1/2$. So, by \Cref{fact:sqrt}, $|\sqrt{1+y}-1| \le \frac{y}{2}+\frac{y^2}{2} \le y,$ which means $\left(\sqrt{1+\frac{\eps}{\sigma^2} d_x} - 1\right)^2 \le (\frac{\eps}{\sigma^2} d_x)^2 = \frac{\eps^2}{\sigma^4} \cdot d_x^2$.
    Therefore, we can bound the right-hand side of \eqref{eq:hellinger-bash-1} as at most
\[\sum_{x \in S} p_x^D \cdot \frac{\eps^2}{\sigma^4}  d_x^2 = \frac{\eps^2}{\sigma^4} \cdot \sigma^2 = \frac{\eps^2}{\sigma^2}.\]
    
    By tensoring properties of the dot product (\Cref{fact:hellinger_product}), we conclude that if we consider an independent sample of $T = \sigma^2/(10\eps^2)\cdot \log 1/\delta$ elements from $D$ or $D'$, the dot product between the corresponding product distributions satisfies 
    \begin{align*}
        H^2(D^{\otimes T}, {D'}^{\otimes T}) \le 1-\left(1-\frac{\eps^2}{\sigma^2}\right)^{T} \le 1-2 \sqrt{\delta}.
    \end{align*}
    The above inequality implies via \Cref{fact:hellinger_implies_l1} that the total variation distance between $D^{\otimes T}$ and ${D'}^{\otimes T}$ is at most $1-\delta$, which rules out a $(1-\delta)$-distinguisher algorithm via  \Cref{fact:l1}.

\paragraph{Distributions with infinite support}
Finally, we discuss how to reduce the case of general distributions to the case of distributions with finite support. We define a distribution $D_r$ by rounding $D$ to the nearest multiple of $\eps/5$ and we denote by $r$ the function that performs this rounding.  By the above proof, there exists a distribution $D_r'$ such that $(1)$ we cannot distinguish $D_r$ and $D_r'$ with probability $1-\delta$ using a smaller number of samples than the claimed lower bound, which we will denote by $LB$, and $(2)$ $|E[D_r] - E[D_r']| \geq 2\eps$.

We now define a distribution $D'$ by first sampling $y = D_r'$ and then sampling and returning a value $X$ from the conditional distribution $X \sim D | r(X) = y$. We now claim that $(a)$ $|E[D] - E[D']| \geq \eps$ and $(b)$ we cannot distinguish $D,D'$ using $o(LB)$ samples. This will imply the theorem.

The rounding changes each value by at most $\eps/10$ and it thus also cannot change the expectation by more than $\eps/10$. Since $r(D) = D_r$, $r(D') = D_r'$ and $|E[D_r] - E[D_r']| \geq 2\eps$, the claim $(a)$ follows.

Suppose that we can distinguish $D,D'$ using $k = o(LB)$ samples. We use this distinguisher to also distinguish $D_r,D_r'$ using $k$ samples. Consider the $k$ samples from either $D_r$ or $D_r'$. For each of the samples $Y$, define a sample $X \sim D | r(X) = Y$. If $Y \sim D_r$ then $X \sim D$, and if $Y \sim D_r'$ then $X \sim D'$. This allows us to use the distinguisher for $D,D'$ to also distinguish with probability $1-\delta$ the distributions $D_r,D_r'$ using $k$ samples. Since we know this is not possible, it must also not be possible to distinguish $D,D'$, proving claim $(b)$.
\end{proof}



\subsection{Instance-optimal I/O-efficient mean estimation}
\label{sec:mean_block}

We now discuss how our upper and lower bounds \Cref{cor:counting_ub,thm:counting_lb} yield (order-oblivious) instance optimal algorithms. 
First, we note that our results yield instance optimal algorithms in the sequential estimation setting. 

\begin{theorem}
    \label{thm:mean_dist}
    There are instance optimal algorithms for problems estimating means and frequencies, as well as their \textsc{Step-Mixture} variants (formally, problems \mean, \stepmean, \freq, \stepfreq) in the sequential estimation setting. 
\end{theorem}
\begin{proof}
    For the \mean problem, this follows from \Cref{cor:counting_ub,thm:counting_lb}. The \freq problem is a special case of \mean where the support of the input distribution is $\{0,1\}$. 

    We now argue that there also exist instance-optimal algorithms for the \textsc{Step-Mixture} versions of the two problems. Namely, we claim that we may simply use the algorithm for \mean/\freq to solve these problems by first computing the mean of each sampled distribution and feeding these means into the respective algorithm. Suppose this algorithm is not instance-optimal -- that is, we may for any $C> 0$ find an algorithm $A$ and a distribution $D$ on which $A$ is $C$ times more efficient than the algorithm we described. We now use this to construct an algorithm for \mean/\freq that is $C$ times faster than our instance-optimal algorithm on some distribution, which is a contradiction.
    
    Let $D'$ denote the distribution of the mean of a distribution sampled from $D$. For $d' \sim D'$, we may sample $d$ from the conditional distribution $d \sim D | d = d'$ and feed this into $A$. It holds that the distribution of $d$ is $D$, meaning that this algorithm has the performance of $A$ on $D$. At the same time, the algorithm is correct by the law of total expectation, which implies that $E[D'] = E[M(D)]$. We thus have a correct algorithm that is $C$ times more efficient on some distribution than our instance-optimal algorithm, as we wanted to.
    %
\end{proof}

Second, we note that these results can be extended to the I/O estimation model. 

\begin{theorem}
\label{thm:mean_io}
    There are order-oblivious instance optimal algorithms for estimating means and frequencies (formally, the \bmean and the \bfreq problem)
    in the I/O estimation model. 
\end{theorem}
\begin{proof}
    This follows from \Cref{thm:mean_dist} by applying the equivalence between the two settings \Cref{thm:io_transfer}. This theorem can be applied because both the \stepmean and the \stepfreq problem satisfy the technical requirements from \Cref{def:consistent,def:bounded}. 
\end{proof}

\subsection{Learning Distributions in $\ell_\infty$} \label{subsec:histogram}

Next, we show how to use our mean estimation algorithm to get a near-instance-optimal algorithm for learning a mixture distribution with respect to the $\ell_\infty$ norm -- formally this is the \minfty problem. 
We then show that using the equivalence between the two setting, we can use this to get an order-oblivious instance-optimal algorithm for estimating a histogram of an input.

Let us now set up a notation for the \minfty problem. We assume each sample $X$ from the input distribution $D$ is itself a distribution over a discrete set $[s] = \{1, 2, \dots, s\}$. Formally, we view each sample $X$ as a sequence of $s$ numbers $p_X(1), p_X(2), \dots, p_X(s)$ with each $p_{X}(j) \ge 0$ and $\sum_{j = 1}^s p_{X}(j) = 1$. The goal is to estimate the probabilities of the mixture distribution $M(D)$; that is, for each $j$ we want to estimate the value $p_{D}(j) := \E_{X \sim D} [p_{X}(j)]$ and the error of our estimate is measured in $\ell_\infty$.  We will use the notation $\mu_j, \sigma_j$ for the mean and variance of $p_{X}(j)$ for $X \sim D$; note that this is a random variable.

We next show how one can get a near-instance-optimal algorithm for the \minfty problem by reducing it to the \mean problem we can already solve. The reduction is done via \Cref{alg:histogram} that is analyzed next.

\begin{algorithm} 
\caption{Algorithm for learning mixtures (\minfty problem)}
\label{alg:histogram}
Sample $T_0 = O(\frac{1}{\eps} \cdot \log \frac{1}{\eps \delta})$ distributions $X_1, X_2, \dots, X_{T_0}$. \\
Let $J_{\text{large}}$ represent the indices $j \in [s]$ where for some $i \le T_0$, $p_{X_i}(j) \ge \eps/2.$ \\
Let $X_{T_0+1}, \cdots$ be samples from $D$\\
For each $j \in J_{\text{large}}$, run the instance optimal algorithm for \mean from \Cref{cor:counting_ub} on $p_{X_{T_0+1}}(j), p_{X_{T_0+2}}(j), \dots$ with parameters $\eps'=\eps$,  $\delta' = O\left( \frac{\eps^3 \cdot \delta}{\log(1/\eps \delta)}\right)$. \\
When the executions of the above algorithm finish, return the outputs as estimates of $p_{D}(j)$ for $j \in J_{\text{large}}$, and return $0$ as the estimate for all $j \not\in J_{\text{large}}$.
\end{algorithm}

\begin{theorem} \label{thm:histograms}
For any input distribution $D$, Algorithm~\ref{alg:histogram} returns estimates $\hat{p}_j, 1\le j \le s$ such that with probability at least $1-\delta$, we have $|p_{M(D)}(j) - \hat{p}_j| \le \eps$ for each $1 \le j \le s$. 

The algorithm has sample complexity
\[
O\left( \left(\frac{1}{\eps} + \max_{1 \le j \le s} \frac{\sigma_j^2}{\eps^2}\right) \cdot \log \frac{1}{\eps \delta} \right).
\]
\end{theorem}

Before we prove Theorem \ref{thm:histograms}, we prove the following lemma.

\begin{lemma} \label{lem:histogram-auxiliary}
    In \Cref{alg:histogram}, the set $J_{\text{large}}$ has size at most $O\left(\frac{1}{\eps^2} \cdot \log \frac{1}{\eps \delta}\right)$. Moreover, with failure probability at most $\frac{\delta}{100}$, any $j$ with $p_{M(D)}(j) \ge \eps$ is in $J_{\text{large}}$.
\end{lemma}

\begin{proof}
    Since $\sum_{j \in [s]} p_{X_i}(j) = 1$ and every $p_{X_i}(j) \ge 0$, at most $\frac{2}{\eps}$ values $j$ can satisfy $p_{X_i}(j) \ge \eps/2$ for any fixed $i$. Thus, at most $\frac{2}{\eps} \cdot T_0 = O\left(\frac{1}{\eps^2} \cdot \log \frac{1}{\eps \delta}\right)$ indices can be in $J_{\text{large}}$.

    Next, consider any $j$ with $p_{M(D)}(j) \ge \eps$. Note that for such $j$ we then have that the random variable $p_{X}(j)$ which is between $0$ and $1$ has a mean at least $\eps$. Hence, with probability at least $\eps/2$ it has to be that $p_{X}(j) \ge \eps/2$ (otherwise its mean would be less than $\eps$). 
    So, for each such $j$, the probability that $X_i$ for some $i$ has $p_{X_i}(j) \ge \eps/2$ is at least $1-(1-\eps/2)^{T_0} \ge 1-\frac{\eps \delta}{100}$. Taking a union bound over all $j$ with $p_{M(D)}(j) \ge \eps$ (of which there can be at most $1/\eps$), we get an overall failure probability of $\delta/100.$
\end{proof}

We now prove \Cref{thm:histograms}.

\begin{proof}[Proof of \Cref{thm:histograms}]
We start by proving the sample complexity is as claimed.
Clearly, the first line satisfies the sample complexity bound.
Now, for each item $j \in J_{\text{large}}$, let $T_j$ represent the sample complexity of \Cref{alg:counting_ones} for $j$. By Lemma~\ref{lem:prob_amplification}, it holds that $\|T_j\|_{O(\log (1/\eps \delta))} \leq O\left(\left(\frac{1}{\eps} + \frac{\sigma_y^2}{\eps^2}\right) \log \frac{1}{\eps \delta}\right)$.
The sample complexity of the algorithm (except the first line) is $\max_{j \in J_{\text{large}}} T_j$.
At the same time, by \Cref{lem:histogram-auxiliary}, it holds that $|J_{\text{large}}| \le O\left(\frac{1}{\eps^2} \log \frac{1}{\eps \delta}\right)$ and thus $\log |J_{\text{large}}| \le O(\log \frac{1}{\eps \delta})$.
By \Cref{lem:expectation_of_max}, we thus have that
\begin{align*}
E[\max_{j \in J_{\text{large}}} T_j] &= O\left(\left(\frac{1}{\eps} + \max_{j \in J_{\text{large}}} \frac{\sigma_j^2}{\eps^2}\right) \log \frac{1}{\eps \delta}\right) = O\left(\left(\frac{1}{\eps} + \max_{k \in [s]} \frac{\sigma_j^2}{\eps^2}\right) \log \frac{1}{\eps \delta}\right) \,,
\end{align*}
as we wanted to prove.

We now argue correctness. First, consider the choices of $j \not\in J_{\text{large}}$. By \Cref{lem:histogram-auxiliary}, with failure probability $\frac{\delta}{100},$ every such $j$ has $p_D(j) < \eps$. Therefore, since we estimate every such probability as $0$, they are all sufficiently accurate.

Next, we consider $j \in J_{\text{large}}$. We have for every such $j$, 
 that we get accuracy $\eps$ with failure probability $O\left( \frac{\eps^3 \delta}{\log(1/\eps \delta)} \right)$ (since this is how we have set the parameters of the subroutine). Since $|J_{\text{large}}| \le O\left(\frac{1}{\eps^2} \log \frac{1}{\eps \delta}\right)$ by \Cref{lem:histogram-auxiliary}, a union bound implies the algorithm is correct for all $j \in J_{\text{large}}$ with probability at least $1-\frac{\delta}{2}$ provided that we choose the constant in the definition of $\delta'$ small enough. Adding the probability $\delta/100$ that the set $J_{\text{large}}$ is wrong, and combining with our bound for $j \not\in J_{\text{large}}$, this completes the proof.
\end{proof}

\begin{theorem} \label{thm:histograms_opt}
In the sequential estimation setting, there is an algorithm for learning a (mixture) distribution w.r.t.\ $\ell_\infty$ (formally, the problems \stepinfty, \distinfty) that is instance optimal up to a factor of $1+\frac{\log \epsilon^{-1}}{\log \delta^{-1}}$. 
\end{theorem}
\begin{proof}
    We start with the problems \stepinfty and \distinfty. As the upper bound, we use the algorithm from \Cref{thm:histograms}. 
    We need to prove a corresponding lower bound. We first observe that for each $j$, solving \stepinfty for some distribution $D$ over a set $S$ enables us in particular to estimate the value of $p_{M(D)}(j)$ up to additive error $\eps$ with error probability $\delta$. That is, we can then also solve the \mean problem on the set $S'$ defined as the set of all possible values of $p_{X}(j)$ for $X \in S$.    A distribution $X$ may give to any item $j$ the probability $p_X(j) = 0$ or $p_X(j) = 1$, meaning that $\min(S') = 0$ and $\max(S') = 1$. 
    Applying the lower bound of \Cref{thm:counting_lb} on this problem, we thus conclude that for any $j \in [s]$, any correct algorithm for \stepinfty on $D$ needs at least $\Omega\left( \left( 1/\eps + \sigma_j^2/\eps^2\right) \cdot \log \delta^{-1}\right)$. This matches (for some $j$) our upper bound up to a factor of $\log(\eps\delta)^{-1} / \log \delta^{-1} = 1 + \log \eps^{-1} / \log \delta^{-1}$. 
    The same argument can be also carried out for the \distinfty problem. 
%
\end{proof}

\begin{theorem}
    \label{thm:histograms_io}
    Consider the problem in the I/O estimation model, where for each element in the alphabet $a \in S$, we want to estimate the frequency of $a$ in the input, and error is measured as $\ell_\infty$ (formally the \bdistinfty problem).
    There is an algorithm in this setting that is optimal up to a factor of $1+\frac{\log \epsilon^{-1}}{\log \delta^{-1}}$.
\end{theorem}
\begin{proof}
    We can solve \stepinfty instance-optimally by \Cref{thm:histograms_opt}.
    We then get the desired algorithm by combining this result with the equivalence between the two settings \Cref{thm:io_transfer}. This can be done since the technical conditions \Cref{def:consistent,def:bounded} are satisfied. 
\end{proof}

\subsection{Estimating mean with multiplicative error}
\label{sec:mean_multiplicative}
We next show how we can get an instance optimal algorithm for estimating the mean with multiplicative error, i.e., the \mean problem with multiplicative error, by reducing it to the additive version of the problem, which we solved above.
Their algorithm (see \cite[Stopping rule algorithm, Approximation algorithm $\fA\fA$]{dagum2000optimal}) is similar to our multiplicative mean estimation algorithm: They also first compute a rough estimate of the mean and the variance and use those rough estimates to compute the number of samples needed for the final estimate.


We first prove the following lemma that enables us to get a rough multiplicative estimate of $\mu$. 
\begin{lemma}
\label{lem:mult}
Consider any distribution $D$ supported on $[0,1]$ with mean $\mu$ and variance $\sigma^2$. There is an algorithm that returns a multiplicative $2$-approximate estimator of $\mu$ with probability at least $1-\delta$ such that its expected sample complexity is $O(\frac{\log 1/\delta}{\mu})$. 
\end{lemma}
\begin{proof}
We estimate $\mu$ by a simple reduction. We perform randomized rounding of the random variable $X \sim D$ to the set $\{0,1\}$, meaning we return 1 with probability $X$ and return 0 otherwise. Note that by the law of total expectation, this does not change the mean of the distribution. This reduced the problem of estimating $\mu$ to the problem of estimating the parameter of a Bernoulli trial. This is a well-known problem with a very simple solution \footnote{The algorithm is that one makes coinflips until we get a constant number of 1's, after which it returns the standard estimate of the mean.} with desired complexity. See the note \cite{watanabe2005sequential} for the argument, which is a special case of the more general algorithm by \citet{lipton1993efficient}.
\end{proof}

\begin{theorem}
    \label{thm:mean_dist_mult}
    In the sequential estimation setting, there is an instance optimal algorithm for learning the mean with multiplicative error (formally, the \mean problem with multiplicative error) over any set $S$ with $\inf(S) > 0, \sup(S) = 1$. 
    
     In the I/O estimation model, there is an order-oblivious instance optimal algorithm for estimating the frequencies with multiplicative error $1+\eps$ (formally the problem \bmean with multiplicative error). 
\end{theorem}

\begin{proof}
        We start with the \mean problem. We note that for any $S \subseteq [0, 1]$, we can first run the algorithm from \Cref{lem:mult} to get an estimate $\bar\mu$ of the mean $\mu$ that is correct up to a multiplicative $2$-factor. Next, we set $\eps' = \eps\bar\mu/10$ and run the estimation algorithm from \Cref{cor:counting_ub} with $\eps'$ to get an estimate of $\mu$ that is correct up to additive $\eps'$, i.e., it is correct up to multiplicative $(1\pm\eps)$. The total expected sample complexity of the algorithm on a distribution $D$ is 
        \begin{align*}
        \Theta\left( \left( \frac{1}{\eps \mu(D)} + \frac{\sigma(D)^2}{(\eps \mu(D))^2} \right) \cdot \log \frac{1}{\delta} \right).
        \end{align*}
This complexity is finite due to our assumption $\inf(S) > 0$. On the other hand, given any distribution $D$, we may apply the lower bound of \Cref{thm:counting_lb} with $\eps' = \eps \mu(D)$ to conclude that any correct algorithm for $D$ with $1\pm\eps$ multiplicative error has at least that complexity.  
    
For the I/O version of the problem, we can make the same reduction as in \Cref{thm:mean_dist,thm:mean_io}. 
\end{proof}

\paragraph{Example: checking whether input contains \texttt{1}} Let us now discuss an example that does not quite fit our equivalence between the two setting (which holds under some mild technical assumptions, \cref{thm:io_transfer}). Imagine the problem in the I/O estimation problem with $B = 1$ where the input consists of a \texttt{0}/\texttt{1} string and our task is to answer with a constant probability of success whether the input contains \texttt{1}. \footnote{Very similar problem could be that the input consists of two strings of the same length and the input block at position $i$ contains the $i$-th letter of both strings.} The natural algorithm for this problem is to read the blocks in random order until we sample \texttt{1}, or until the whole string is read; it is not hard to prove that this algorithm is order-oblivious instance optimal.

In the sequential estimation setting, the corresponding problem is to estimate whether $p > 0$ or $p = 0$ for an input Bernoulli distribution $B(p)$. Again there is an analogous estimation algorithm that is instance optimal, except on input $B(0)$ where the algorithm does not terminate and thus is not valid on that input (no valid algorithm can both terminate on $B(0)$ and be correct). 

We note that this problem is a variant of our multiplicative mean estimation problem when we try to estimate the mean up to any multiplicative factor. Indeed, both above instance optimal algorithms are a variant of the algorithm from \cref{lem:mult}. On the other hand, we notice that \cref{thm:mean_dist_mult} is not directly applicable to the problem since $\inf(S) = 0$. Also, we cannot use the equivalence between the two settings \cref{thm:io_transfer}, since the distance function is not uniformly consistent per \cref{def:consistent}. This makes sense since the two variants of the algorithm for the I/O estimation and sequential estimation are slightly different: In the I/O setup, the algorithm can terminate after it read the whole input. On the other hand, a sequential estimator can never terminate, since it could be the case that the probability of \texttt{1} is simply very small. 

\section{Quantiles}
\label{sec:quantiles}
In this section, we first present an approximately instance-optimal algorithm for the problem of estimating a quantile of a mixture distribution, i.e., the \quantile problem and its variants. 
This allows us to construct approximately instance optimal algorithms for variants of the problem of learning distributions with respect to the Kolmogorov-Smirnov norm, i.e., the \distks problem in \Cref{sec:quantiles_ks}. 

Recall that that $Dist(S)$ denotes the set of distributions over a ground set $S$. 
In the \mquantile problem, we consider a distribution $D$ over $Dist(S)$ which induces a mixture distribution $M(D)$. 
We have sampling access to $D$ and the problem is parametrized by a number $q \in [0, 1]$. Our task is to estimate the $q$-th quantile $Q_{M(D)}(q)$ of $M(D)$ defined as $$Q_{M(D)}(q) = \inf \{x \,|\, P[M(D) > x]  \geq q\}. $$

We define $q^-_X(x) = P[X < x]$ and $q^+_X(x) = P[X \leq x]$. For continuous distributions, we define $q_X(x) = q_X^-(x) = q_X^+(x)$. 
We want to estimate $Q_{M(D)}(q)$ with $\eps$ accuracy, meaning that we should output a value $\hN$ such that
\[
q^-_{M(D)}(\hN)-\eps \leq q \leq q^+_{M(D)}(\hN)+\eps \,.
\]

%
%

We need the following definition in order to be able to specify the instance-optimal sample complexity.
Let\footnote{Recall, that $\otimes$ denotes the product distribution.} $N = Q_{D \otimes Unif[0,1]}(q)$. We define the following quantity:
\[
\sigma_q^2 = \Var_{X \sim D}(q_{(X,U)}(N)) 
\]
for $U \sim Unif[0,1]$.

\begin{theorem}
\label{thm:quantile_main}
\Cref{alg:quantiles_final} is up to $O\left(1+\frac{\log \left( \delta^{-1} \log \eps^{-1} \right)}{\log \delta^{-1}}\right)$ an instance-optimal algorithm for the \mquantile problem. Let $d^-_X = |q-q^+_X(Q_D(q))|$ and $d^+_X = |q - q^-_X(Q_D(q))|$ and let  $\bar{\eps} = \max(\eps, \min(d^-, d^+))$. Its sample complexity is
\begin{align*}
    \Theta\left( \left( \frac{1}{\bar{\eps}} + \frac{\sigma_q^2}{\bar{\eps}^2} \right) \cdot \frac{\log \left( \delta^{-1} \log \eps^{-1} \right)}{\log \delta^{-1}} \right). 
\end{align*}
\end{theorem}
We split the proof into an upper bound (\Cref{thm:quantile_ub_final}) and a lower bound (\Cref{thm:quantile_lb}).
We first prove the upper bound. We do this by first giving a suboptimal algorithm in \Cref{sec:quantiles_upper} and then in \Cref{sec:quantiles_upper_final}, we give a black-box transformation, that makes it optimal.

We remark that the factor $O\left(1+\frac{\log \left( \delta^{-1} \log \eps^{-1} \right)}{\log \delta^{-1}}\right)$ in \cref{thm:quantile_main} cannot be improved. We prove this in \cref{sec:gap}. 

We observe that via \Cref{thm:io_transfer} we also get approximately instance optimal algorithms for the following problems. 

\begin{theorem}
\label{thm:instance_optimal_median_general}
    Up to a $1 + \frac{\log \left( \delta^{-1} \log \eps^{-1} \right)}{\log \delta^{-1}}$ factor, there is an instance optimal algorithm for estimating a quantile of a (step-mixture) distribution (formally, the problems \stepquantile, \quantile). 
    Moreover, up to the same factor, there is an order-oblivious instance optimal algorithm for estimating a quantile in the I/O estimation model (formally the problem \bquantile).  
\end{theorem}
\begin{proof}
    We get the result for \stepquantile from \Cref{thm:quantile_ub_final} and \Cref{thm:quantile_lb}. \quantile is a special case of \stepquantile for $B=1$. Since \stepquantile satisfies the technical conditions of \Cref{def:consistent,def:bounded}, we can apply  \Cref{thm:io_transfer} to get the result in the I/O setting. 
\end{proof}

\subsection{Analysing \Cref{alg:quantiles}, the core subroutine}
\label{sec:quantiles_upper}

In this section, we focus on what we call mixture-continuous distributions. We later remove this assumption using a ``continualization trick''. We say that a distribution is called mixture-continuous if the probability it samples a non-continuous distribution is non-zero.
We now analyze \Cref{alg:quantiles}; specifically, we prove the following theorem:
\begin{theorem}
\label{thm:quantile_ub}
Assume $D$ gives zero probability to any singleton event.
In expectation, \Cref{alg:quantiles} uses $O(1/\eps + \sigma^2_q/\eps^2)$ samples from the input distribution $D$.
Moreover, with constant probability, it outputs $\hN$ which is an $\eps$-correct estimate of the $q$-th quantile of $M(D)$.
\end{theorem}

\begin{algorithm}
$T_1 \leftarrow O(1/\eps)$\\
$X_1, \cdots, X_{T_1} \leftarrow$ get $T_1$ samples\\
$\tN \leftarrow $ $q$-th quantile of $M(X_1, \cdots, X_{T_1})$\\
$X_{T_1+1}, \cdots, X_{2T_1} \leftarrow$ get $T_1$ samples\\
$\tsigma^2 \leftarrow \frac{1}{T_1} \sum_{i = T_1+1}^{2T_1} \left( q_{X_i}(\tN) - q \right)^2$\\
$T_2 \leftarrow O\left(1/\eps + \tsigma^2/\eps^2\right)$\\
$X_{2 T_1 + 1}, \cdots, X_{2 T_1 + T_2} \leftarrow$ get $T_2$ samples\\ 
$\hN \leftarrow $ $q$-th quantile of $M(X_{2T_1+1}, \cdots, X_{2T_1 + T_2})$\\
\Return{$\hN$}
\caption{Three-phase Quantile Estimation} \label{alg:quantiles}
\end{algorithm}

We start by proving several lemmas.
\begin{lemma}
\label{lem:continuity}
Let the input distribution $D$ be mixture-continuous.
For any $q, \beta > 0$ with $q+\beta \le 1$, we have $|\sigma_{q + \beta}^2 - \sigma_{q}^2| \le 5 \beta$. 
\end{lemma}

\begin{proof}
For any $q'$, define $N_{q'} = Q_{M(D)}(q')$. 
For every distribution $X \in Dist$, we define $\delta(X) := q_{X}(N_{q+\beta}) - q_X(N_q)$, i.e., the probability mass $X$ between the quantiles $q$ and $q+\beta$ of $M(D)$. 
Note that we have by definition
\begin{equation}
    \label{eq:rem}
    \E[\delta(X)] = \beta
\end{equation}
where we have used the assumption of $D$ being mixture-continuous, which implies that $M(D)$ is continuous. We can then write
\begin{align} \label{eq:sigma-q+eps}
    \sigma_{q+\beta}^2
    &= \E_{X \sim D}[((q_X(N_{q+\beta}) - (q+\beta))^2]\\
    &= \E_{X \sim D}[((q_X(N_{q}) - q + \delta(X) - q)^2]\\
    &= \sigma_q^2  + \E_{X \sim D}[(\delta(X) - \beta)^2] + 2 \E_{X \sim D}[(q_X(N_q) - q)(\delta(X) - \beta)]
\end{align}
In view of \Cref{eq:rem}, the middle term is nonnegative, and is maximized if with probability $\beta$, we have $\delta(X) = 1$ and otherwise $\delta(X) = 0$; thus we conclude that
\begin{align}\label{eq:middle-term}
\begin{split} 
    0 &\le \E_{X \sim D}[(\delta(X) - \beta)^2]
    \le \beta (1-\beta)^2 + (1-\beta) \beta^2
    \le \beta. 
    \end{split}
\end{align}
We continue with the last term. By \Cref{eq:rem}, and since $|q_X(N_q)-q| \le 1$, we have
\begin{align}
& \left|2 \E_{X \sim D}[(q_X(N_q) - q)(\delta(X) - \beta)]\right| \leq 2 \E_{X \sim D}[|\delta(X) - \beta|] \leq 2 \E_{X \sim D}[\delta(X) + \beta] = 4\beta\label{eq:last-term}
%
%
\end{align}
By combining Equations \eqref{eq:sigma-q+eps}, \eqref{eq:middle-term}, and \eqref{eq:last-term}, the proof is complete.
\end{proof}

\begin{claim}
\label{cl:Ngood}
Let $D \in Dist(S)$ be a mixture-continuous distribution and let $\tq = q_{M(D)}(\tN)$. Then, with probability $0.99$,
\[
|\tq -q| \leq O\left(\frac{\sigma_q}{\sqrt{T_1}}+\frac{1}{T_1}\right) = O\left(\sigma_q \cdot \sqrt{\eps} + \eps\right).
\]
Moreover, $\E[||\tilde{q}-q|] \leq O\left(\sigma_q \cdot \sqrt{\eps} + \eps\right)$.
\end{claim}
\begin{proof}
For any $\beta > 0$, we will bound the probability of $\tq < q - \beta$ (the same argument can then be done for the probability of $\tq < q + \beta$). We will then determine how to appropriately set $\beta$. We will also prove bounds on the tail probabilities, which we will use in order to bound the expectation $\E[|\tilde{q}-q|]$.

Note that $\tq < q - \beta$ implies that
\begin{align*}
\sum_{i = 1}^{T_1} q_{X_i}(N_{q-\beta}) \geq q T_1
\end{align*}
Because every $X_i$ is drawn independently, the left-hand side is the sum of $T_1$ independent variables, each with mean $q-\beta$ and variance $\sigma_{q-\beta}^2 \le \sigma_q^2+5 \beta$ (by \Cref{lem:continuity}), and bounded between $0$ and $1$. Therefore, by Bernstein's inequality, 
\begin{align*}
  \iffocs & \fi  P\left(\sum_{i = 1}^{T_1} q_{X_i}(N_{q-\beta}) \ge qT_1\right)
   \iffocs \\ \fi&= P\left(\sum_{i = 1}^{T_1} q_{X_i}(N_{q-\beta}) - T_1 (q-\beta) \ge T_1 \beta\right) \\
    &\le \exp\left(- \Omega\left( \min\left(T_1 \beta, \frac{T_1 \beta^2}{\sigma_{q-\beta}^2} \right)\right) \right) \\
    &\le \exp\left(- \Omega\left( \min\left(T_1 \beta, \frac{T_1 \beta^2}{\sigma^2+5\beta} \right)\right) \right) \\
    &= \exp\left(- \Omega\left( \min\left(T_1 \beta, \frac{T_1 \beta^2}{\sigma^2} \right)\right) \right).
\end{align*}

    Hence, if we set $\beta = K (\sigma/\sqrt{T_1} + 1/T_1)$, for any parameter $K \ge 1$, this implies that the probability that $\tilde{q} < q-\beta$ is at most $\exp\left(-\Omega(K)\right)$.
  Setting $K=1$, we get the first half of the claim since $T_1 = \Theta(1/\eps)$.

  We now use the above tail bounds in order bound the expectation using a standard calculation:
  {\iffocs\small\fi
  \begin{align*}
  &\E[|\tilde{q} - q|] \\&= \int_0^\infty P[|\tilde{q} - q| \geq z]\,dz \\&\leq \sum_{K=0}^\infty (\sigma/\sqrt{T_1} + 1/T_1) P[|\tilde{q} - q| \geq  K (\sigma/\sqrt{T_1} + 1/T_1)] \\&\leq  (\sigma/\sqrt{T_1} + 1/T_1) \sum_{K=0}^\infty  \exp(-\Omega(K)) \\&= O(\sigma/\sqrt{T_1} + 1/T_1) \,.
  \end{align*}
  }
  Again, since $T_1 = \Theta(1/\eps)$, this proves the claim.


\end{proof}

\begin{claim}
\label{cl:sigmagood}
It holds
\begin{align*}
    \tsigma^2 = \sigma_q^2 \pm O(\sqrt{\eps} \sigma_q + \eps)
\end{align*}
with probability at least $0.98$. 
\end{claim}
\begin{proof}
We condition on a fixed $\tilde{N}$ and $\tilde{q}$, and assume $\tq \in q \pm O(\sqrt{\eps}\sigma+\eps)$. Note that by \Cref{cl:Ngood}, this happens with probability $0.99$.

Conditioned on $\tilde{q}$, we have for $X \sim D$ that $q_X(\tN)$ is a random variable with mean $\tilde{q}$ and variance $\sigma_{\tilde{q}}^2$. Moreover, if we write $Y = (q_{X}(\tN)-q)^2$, note that $Y$ is bounded between $0$ and $1$, and thus has variance at most $\E[Y^2] \le \E[Y].$ We thus have 
\[\Var[Y] \le \E[Y] = \E[(q_X(\tN)-q)^2] = \sigma_{\tilde{q}}^2 + (\tilde{q}-q)^2.\]
Under the assumption that $\tilde{q} = q \pm O(\sqrt{\eps} \sigma + \eps)$ and by \Cref{lem:continuity}, we have that $0 \le (\tilde{q}-q)^2 \le |\tilde{q}-q| \le O(\sqrt{\eps} \sigma + \eps),$ and $|\sigma_{\tilde{q}}^2-\sigma^2| \le 5 |\tilde{q}-q| \le O(\sqrt{\eps} \sigma + \eps)$. Therefore, we may bound
\[\sigma^2 - O(\sqrt{\eps} \sigma + \eps) \le \sigma_{\tilde{q}}^2 + (\tilde{q}-q)^2 \le \sigma^2 + O(\sqrt{\eps} \sigma + \eps).\]

In summary, under our assumption that $\tilde{q} = q \pm O(\sqrt{\eps} \sigma + \eps)$, we have that $\E[(q_X(\tN)-q)^2|\tilde{N}] = \sigma^2 \pm O(\sqrt{\eps} \sigma+\eps)$, and $Var[(q_X(\tN)-q)^2|\tilde{N}] \le \sigma^2+O(\sqrt{\eps} \sigma + \eps) \le O(\sigma^2+\eps)$.
At the same time, $\tilde{\sigma}^2$ is the average of $T_1 = c/\eps$, for some constant $c$, conditionally independent (on $\tilde{N}$) copies of $(\tilde{q}(B)-q)^2$.
By the Chebyshev inequality, we then have
{\iffocs \small \fi
\begin{align*}
 &P\left(\left|\tilde{\sigma}^2-\sigma_q^2 \right| \ge O(\sqrt{\eps} \sigma + \eps) | \tilde{N}\right) \\ &\leq P\left(\left|\tilde{\sigma}^2-\E[\tilde{\sigma}^2|\tilde{N}] \right| \ge O(\sqrt{\eps} \sigma + \eps)-O(\sqrt{\eps} \sigma + \eps)|\tilde{N}\right) \\&\leq \frac{c^{-1} \eps \cdot O(\sigma^2 + \eps)}{O(\sqrt{\eps}\sigma + \eps)^2} \leq 0.01
\end{align*}
}
where the last inequality holds for the constant $c$ being sufficiently large. Adding the probability $0.01$ that $\tq \not\in q \pm O(\sqrt{\eps}\sigma+\eps)$, we get exactly the desired claim.
%
%
\end{proof}

\begin{proof}[Proof of \Cref{thm:quantile_ub}.]
First, we bound the expected number of samples, which is $2T_1+\E[T_2]$. We have $T_1=O(1/\eps)$, and $\E[T_2] = O(1/\eps + \E[\tilde{\sigma}^2]/\eps^2).$ Moreover, since we used fresh samples to determine $\tilde{\sigma}^2$ in line 6 of \Cref{alg:quantiles}, if we condition on $\tilde{q}$, we have that $\E[\tilde{\sigma}^2|\tilde{q}] = \sigma_{\tilde{q}}^2$. Thus, by combining \Cref{lem:continuity} and \Cref{cl:sigmagood}, we have 
\begin{align*}E[\tilde{\sigma}^2|\tilde{q}] &\le \sigma^2+O(\E[|\tilde{q}-q|]) \le O(\sigma^2 + \sigma \sqrt{\eps} + \eps) \le O(\sigma^2+\eps).\end{align*}
Thus, $\E[T_2] \le O(1/\eps + \sigma^2/\eps)$, which proves the desired bound on the number of samples.

Next, we verify the accuracy.
Let us have a constant $C$ large enough.
If $\sigma^2 \le C \eps$, then $T_2 \ge \frac{C^2}{\eps} \ge \frac{C^2}{\eps} \cdot \frac{\sigma^2}{C  \eps} = \frac{C}{K^2} \cdot \frac{\sigma^2}{\eps^2}$. Otherwise, by \Cref{cl:sigmagood}, we have that with probability $0.98$, $\tilde{\sigma}^2 \ge \sigma^2 - O(\sqrt{\eps} \sigma + \eps) \ge \sigma^2/2,$ as long as $C$ is sufficiently large. This means that
\begin{align*}
    T_2 \ge C \cdot \frac{\sigma^2/2}{\eps^2} \ge \frac{C}{2} \cdot \frac{\sigma^2}{\eps^2}.
\end{align*}
Thus, in either case, we conclude that $T_2 \ge \max\left(\frac{C}{\eps}, \frac{C}{2} \cdot \frac{\sigma^2}{\eps^2}\right)$. 
Then we can apply \Cref{cl:Ngood} again, replacing $T_1$ with $T_2$.
Specifically, because we are using fresh samples, \Cref{cl:Ngood} implies that $\hat{q}$, the quantile of $\hat{N}$, must satisfy 
\[\hat{q} = q \pm O\left(\frac{\sigma}{\sqrt{T_2}}+\frac{1}{T_2}\right) = q \pm \eps\]
with probability $0.98$, where the last equality holds for large enough value of $C$.
\end{proof}

\subsection{Getting an instance-optimal algorithm}
\label{sec:quantiles_upper_final}

\begin{algorithm}
\For{$\eps_i = 2^{-1}, 2^{-2}, \cdots, 2^{-\lceil \log(\eps^{-1}) \rceil}$}{
$(\hat{N}_-, U_-) \leftarrow$ execute \Cref{alg:quantiles}, for $q= q-\eps_i/10$ with accuracy $\eps_i/100$ with probability $1-0.1 \delta / \log \eps^{-1}$ on distribution $D \otimes Unif[0,1]$\\
$(\hat{N}_+, U_+) \leftarrow$ execute \Cref{alg:quantiles}, for $q= q+\eps_i/10$ with accuracy $\eps_i/100$ with probability $1-0.1 \delta / \log \eps^{-1}$ on distribution $D \otimes Unif[0,1]$\\
\If{$\hat{N}_- = \hat{N}_+$}{
\Return{$\hat{N}_-$}
}
}
\Return{$\hat{N}_-$}
\caption{Instance-optimal Quantile Estimation} \label{alg:quantiles_final}
\end{algorithm}
Note that we can easily perform samples from $D \otimes Unif[0,1]$ by always sampling from $X \sim D$ and $U \sim Unif[0,1]$ and returning $(X,U)$.

\begin{theorem}
\label{thm:quantile_ub_final}
Let $\bar{\eps} = \max(\eps, \min(d^-, d^+))$. 
In expectation, \Cref{alg:quantiles_final} uses $$O\left((1/\bar{\eps} + \sigma^2/\bar{\eps}^2) \cdot  \log \frac{\log \eps^{-1}}{\delta} \right)$$ samples.
With probability $1 - \delta$, the algorithm outputs $\hN$ that is an additive $\eps$-approximation to the quantile $q$.
\end{theorem}
\begin{proof}
We start by proving correctness.
By union bound, all executions of \Cref{alg:quantiles} will give a correct answer with probability $1 - \delta/5$.
We condition on this being the case. If we return in the last possible round (i.e. $\eps_i = 2^{-\lceil \log \eps^{-1} \rceil} < \eps$), the error is at most $\eps_i/10+\eps_i/100 \leq \eps/10+\eps/100 < \eps$. If we return earlier, then it is because $\hat{N}_- = \hat{N}_+$. At the same time, for $N = Q_{M(D)}(q)$, we have that $\hat{N}_- \leq N \leq \hat{N}_+$. Since the first and last of these values are equal, we get that they are both equal to $N$, meaning that we return a correct answer.

It remains to prove the expected complexity. Let $i^*$ be the first round such that it holds $\eps_i \leq \min(d^-,d^+)$. We now claim that in any round later than $i^*$, we finish with probability $1-0.2 \delta / \log \eps^{-1} \geq 4/5$. This will then allow us to bound the expected sample complexity.

Consider a round in which $\eps_i \leq \min(d^-,d^+)$. Assuming we get a correct answer, the error in $\hat{N}_-$ is $\leq \eps_i/10 + \eps_i/100 < \eps_i \leq d^-$, meaning that $\hat{N}_- = N$. Similarly, we get that $\hat{N}_+ = N$. But this means that $\hat{N}_- = \hat{N}_+$, meaning that we indeed finish in the given round. The probability that the estimates $\hat{N}_-,\hat{N}_+$ are both $\eps_i/100$-correct is then $1-0.2 \delta / \log \eps^{-1} \geq 4/5$ as we wanted.

We are now ready to bound the expected complexity. Let $\tau$ be the iteration in which we return. The expected sample complexity is then equal to 
\begin{align*}
&\sum_{i=1}^{\lceil \log \eps^{-1} \rceil} P[\tau = i] \cdot O\left((2^i + \sigma^2 4^i) \cdot \log \frac{\log \eps^{-1}}{\delta}\right) \\
&\leq \left(\sum_{i=1}^{i^*} O(2^i + \sigma^2 4^i) + \sum_{\mathclap{i=i^*+1}}^{\mathclap{\lceil \log \eps^{-1} \rceil}} 10^{i - i^*} O(2^i + \sigma^2 4^i)\right) \cdot \log \frac{\log \eps^{-1}}{\delta} 
\\&= O\left((2^{i^*} + \sigma^2 4^{i^*})\log \frac{\log \eps^{-1}}{\delta}\right) \\&= O\left((1/\bar{\eps} + \sigma^2/\bar{\eps}^2)\log \frac{\log \eps^{-1}}{\delta}\right). 
\end{align*}
\end{proof}

\subsection{Lower bound}
\label{sec:quantiles_lower}

Here, we show that our quantile sequential estimator is instance-optimal.

\begin{theorem}
\label{thm:quantile_lb}
Let $D, q, \eps, \delta$ be an instance to the quantile estimation problem and $\sigma = \sigma_q(D)$. Let $\bar{\eps} = \max(\eps, \min(d^-,d^+))$. Any correct algorithm needs at least
\begin{align*}
    LB = \Omega\left( \left( \frac{1}{\bar{\eps}} + \frac{\sigma^2}{\bar{\eps}^2} \right) \cdot \log \frac{1}{\delta} \right)
\end{align*}
samples from $D$. 
\end{theorem}
\begin{proof}
It will be easier for us to work with a continuous distribution.
For any distribution $D$, we thus define a distribution $\tilde{M}(D)$ as $(X,U)$ for $X \sim M(D)$ and $U \sim Unif[0,1]$. We compare values from the base set of this distribution lexicographically. We now show a lower bound for estimating a quantile of the first component of this distribution or, equivalently, a lower bound on estimating a quantile of $M(D)$.

Assume without loss of generality that $d^- \leq d^+$. It then holds that $\bar{\eps} = \max(\eps, d^-)$.
We prove the existence of two distributions $D_1, D_2$ such that: (1) they are both indistinguishable from $D$ by fewer than $LB$ samples, and thus also from each other, and (2) they have disjoint sets of correct answers for the $\mquantile$ problem. This will imply the theorem.

Let $N = (N_1,N_2) = Q_q(\tilde{M}(D))$. Let $N^- = (N_1^-,N_2^-) = Q_{q-2\bar{\eps}}(\tilde{M}(D))$. Note that by the way we defined $\bar{\eps}$, it holds $N^-_1 < N_1$.
For $X \sim D$ (note that $X$ is itself a distribution), define $y = q_{(X,U)}(N), x^- = q_{(X,U)}(N^-)$ for $U \sim Unif[0,1]$ and let $D_y,D_{y^-}$ denote the distributions of these two random variables. 
Note that the standard deviation of $y$ is $\sigma_{q}$ and the standard deviation of $y^-$ is $\sigma_q - O(\eps)$ by \Cref{lem:continuity}.

We now construct $D_1$. Let $r$ be the function that rounds a number to the nearest multiple of $\eps$, and let $D_r = r(D_y)$.
We apply \Cref{lem:counting_lb_2} to $D_r$ to get that there exists a distribution $y' \sim D_r'$ such that $\E[y'] \leq q - 3\bar{\eps}$ and such that distinguishing $D_r$ and $D_r'$ requires $LB$ samples. We use this to construct $D_1$:
We first sample $y \sim D_r'$, then we sample and return the value (which in itself is a distribution) $X$ sampled from the conditional distribution $X \sim D | r(q_{X,U}(N)) = y$.\footnote{Note that there are only finitely possible values of $y$, which implies that this is correctly defined. Indeed, this is the reason we perform the rounding as otherwise, we would have issues with measurability.}

We now claim that $q_{\tilde{M}(D_1)}(N) \leq q-2\bar{\eps}$. Note that for $X \sim D_1$, it then holds $r(q_{(X,U)}(N)) \sim y'$. At the same time, we have that $r(q_{\tilde{M}(D_1)}(N)) = \E[y'] \leq q - 3\bar{\eps}$. At the same time, the function $r$ changes any value by at most $\eps/2$, which implies the claim.

We now argue that $D$ and $D_1$ are indistinguishable by fewer than $LB$ samples. We prove that this would imply that $D_r,D_r'$ are also distinguishable by fewer than $LB$ samples. Since this cannot be the case, $D$ and $D_1$ are also indistinguishable. For each sample $Y$ which may be from either $D_r$ or $D_r'$, we sample $X \sim D | r(q_{X,U}(N)) = Y$. If $Y \sim D_r$ then $X \sim D$, if $Y \sim D_r'$ then $X \sim D_1$. This allows us to use a distinguished for $D,D_1$ to also distinguish $D_r,D_r'$.


We now construct $D_2$. Note that we may assume $\eps < \sigma_q/2$ (otherwise $LB = \Omega(1)$), meaning that $\E[x^-] = q - 2 \bar{\eps}$ and the standard deviation of $x^-$ is $\Theta(\sigma_q)$. By the same argument as above, we thus may obtain a distribution $D_2$ which cannot be distinguished from $D$ by fewer than $LB$ samples, and which has $q_{\tilde{M}(D_2)}(N^-) \geq q+2\bar{\eps}$.


Since $q_{\tilde{M}(D_1)}(N) \geq q+2\bar{\eps}$, any correct answer for the $\eps$-approximate quantile problem on the first component of the distribution $D_1$, has to be $\geq N_1$. Similarly, since $q_{\tilde{M}(D_2)}(N^-) \geq q+2\bar{\eps}$, any correct answer for $D_2$ has to be $\leq N_1^-$. Since $N_1^- < N_1$, we thus have that the sets of correct answers for $D_1$ and $D_2$ are disjoint. At the same time, both $D_1$ and $D_2$ cannot be distinguished from $D$ in fewer than $LB$ samples; they thus also cannot be distinguished from each other since any such distinguished could be used to distinguish at least one of $D_1,D_2$ from $D$. This finishes the proof.
\end{proof}

\subsection{Learning mixture distributions}
\label{sec:quantiles_ks}

We remark that the quantile estimation algorithm from \Cref{thm:quantile_main} can be converted into an algorithm for learning the input distribution in Kolmogorov-Smirnov distance (formally, problems \distks and \mks). 
In other words, we can output a distribution $\hat{M}$ such that $KS(M(D),\hat{M}) = \sup_x |P[M(D) \leq x] - P[\hat{M} \leq x]| \le \eps$, using a nearly instance-optimal number of samples.
Specifically, we prove the following result. 

\begin{theorem}
    In the sequential estimation setting, given sampling access to a distribution $D$ on $Dist(\mathbb{R})$, there are algorithms for learning the mixture distribution $M(D)$ in the Kolmogorov-Smirnov norm (the \distks and \mks problems) that are instance optimal up to a factor of 
    $$
    \log \log \eps^{-1} \cdot \left(1 + \frac{\log \eps^{-1}}{\log \delta^{-1}}\right)
    $$
    We similarly get an approximately order-oblivious instance optimal algorithm for learning the input distribution in KS-norm in the I/O estimation setting (formally the problem \bdistks). 
\end{theorem}
\begin{proof}
    We will focus on the \mks. The result for \distks is a special case for $B=1$. The result in the I/O setting follows by \Cref{thm:io_transfer}. 

Our approach is similar to how we convert mean estimation into mixture learning in $\ell_\infty$ (\Cref{subsec:histogram}). The algorithm will generate samples $X_1, X_2, \dots$ from $D$ (modified using \Cref{lem:prob_amplification}), and will simultaneously attempt the quantile estimation for the quantiles $\eps/3, 2\eps/3, \eps, \dots, 1-\eps/3, 1$, each with accuracy $\eps/10$ and failure probability $\frac{\eps \cdot \delta}{3}$. For each $q \in \{\eps/3, \eps, \dots, 1\},$ let $\hat{N}_q$ represent the predicted quantile. Then, we set the distribution $\hat{M}$ as the uniform distribution on our estimates $\hat{N}_q$.

With at most $\frac{\eps \cdot \delta}{3} \cdot \frac{3}{\eps} \le \delta$ failure probability, every $\hat{N}_q$ for $q \in \{\eps/2, \eps, \dots, 1\}$ will be within the $(q-\eps/10)$th and $(q+\eps/10)$th quantiles. This also means that the values $\hat{N}_q$ will be in nondecreasing order.

Now, consider any value $x$, and suppose that $q = P(M(D) \le x)$. Suppose that $r \cdot \eps/3 \le q < (r+1) \cdot \eps/3,$ where $r$ is some nonnegative integer. We will prove that $P(\hat{M} \le x) \ge (r-1) \eps/3$ and $P(\hat{M} \ge x) \le (r+2) \eps/3$, which is sufficient.

We first show that $P(\hat{M} \le x) \ge (r-1) \eps/3$. (We can assume without loss of generality $r \ge 2$.) This is true if $\hat{N}_{(r-1) \eps/3} \le x,$ since it also implies that $\hat{N}_q \le x$ for all $q \le (r-1) \eps/3$ that is a multiple of $\eps/3$. However, we are assuming that every quantile is accurate up to error $\eps/10$, so $\hat{N}_{(r-1) \eps/3}$ is at most the $(\frac{(r-1) \eps}{3} + \frac{\eps}{10}) \le (q - \frac{\eps}{10})$th true quantile in $M(D)$. This is at most $x$.

Next, we show that $P(\hat{M} \le x) \le (r+2) \eps/3$. This is true if $\hat{N}_{(r+2) \eps/3} > x,$ since it also implies that $\hat{N}_q > x$ for all $q \ge (r+2) \eps/3$ that is a multiple of $\eps/3$. Since every quantile is accurate up to error $\eps/10$, $\hat{N}_{(r+2) \eps/3}$ is at least the $(\frac{(r+2) \eps}{3} - \frac{\eps}{10}) \ge (q + \frac{\eps}{10})$th true quantile in $M(D)$. However, since $q$ equals $P(M(D) \le x)$, the $(q + \frac{\eps}{10})$th quantile is strictly larger than $x$. This completes the accuracy proof.

Finally, we need to verify that the number of samples is nearly optimal. Because learning the distribution up to KS distance $\eps$ implies we estimate every quantile up to error $\eps$, by \Cref{thm:quantile_lb}, we must use $\Omega\left(\max_{q \in \{\eps/3, 2 \eps/3, \dots, 1\}} \left( \frac{1}{\bar{\eps}_q} + \frac{\sigma_q^2}{\bar{\eps}_q^2} \right) \cdot \log \frac{1}{\delta} \right)$ samples, where $\sigma_q, \bar{\eps}_q$ are the corresponding values of $\sigma$ and and $\bar{\eps}$ for the $q$th quantile (as in \Cref{thm:quantile_lb}).
Next, by Lemma~\ref{lem:prob_amplification}, it holds that $\|T_q\|_{O(\log (1/\eps \delta))} \leq O\left(\left(\frac{1}{\bar{\eps}_q} + \frac{\sigma_q^2}{\bar{\eps}_q^2}\right) \log\log \eps^{-1} \cdot \log \frac{1}{\eps \delta}\right)$, where $T_q$ is the number of samples needed to estimate the $q$th quantile with probability $\frac{\eps \cdot \delta}{3}$. The overall number of samples is $\max_{q \in \{\eps/3, 2 \eps/3, \dots, 1\}} T_q$, which by \Cref{lem:expectation_of_max} is at most $O(\max_q \|T_q\|_{O(\log(1/\eps \delta))})$. Hence, the sample complexity is at most $\log \log \eps^{-1} \cdot \left(1 + \frac{\log \eps^{-1}}{\log \delta^{-1}}\right)$ times the optimal.
\end{proof}

\section{Experiments}

\begin{figure*}[t]
    \centering
    \includegraphics[width = 0.8 \textwidth]{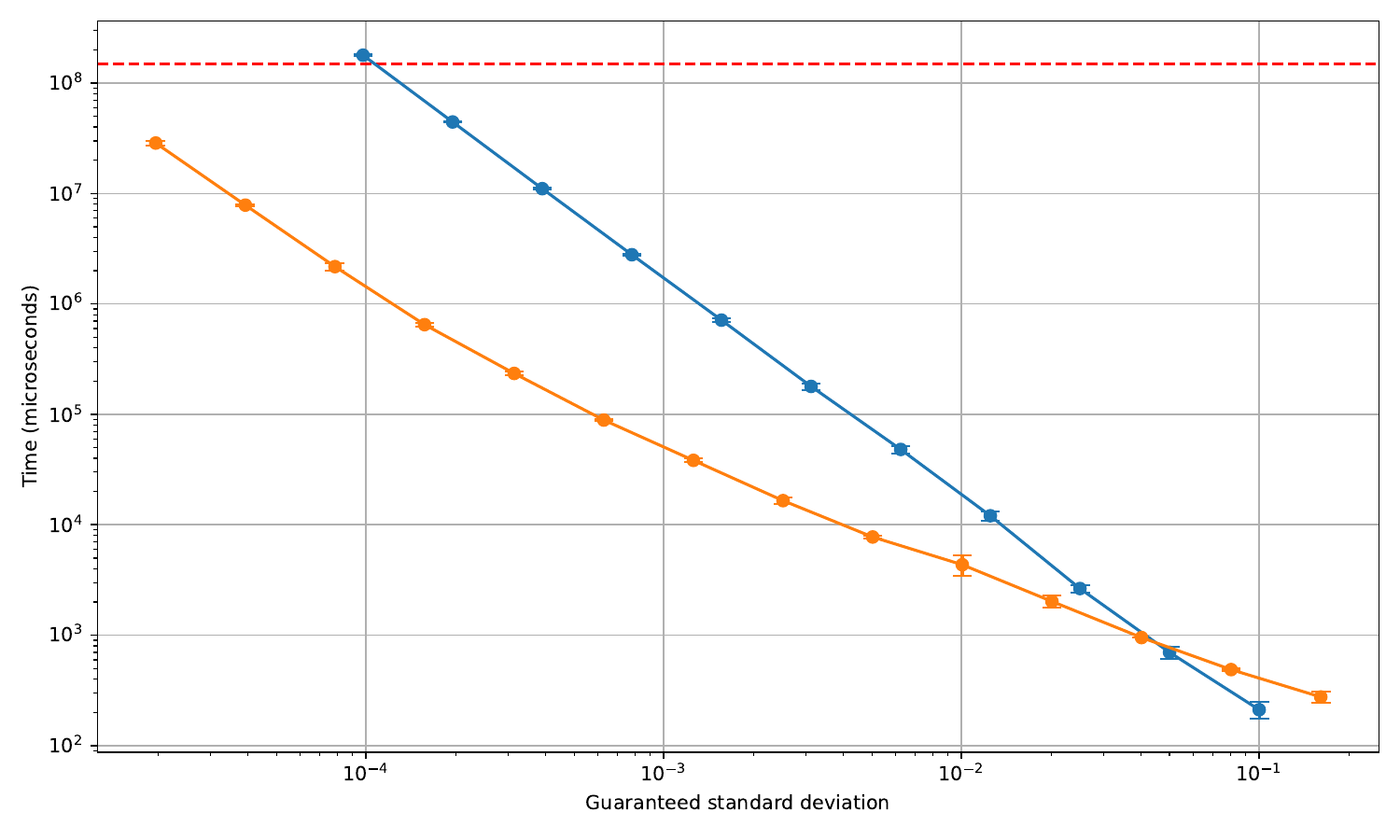}
    \includegraphics[width = 0.8 \textwidth]{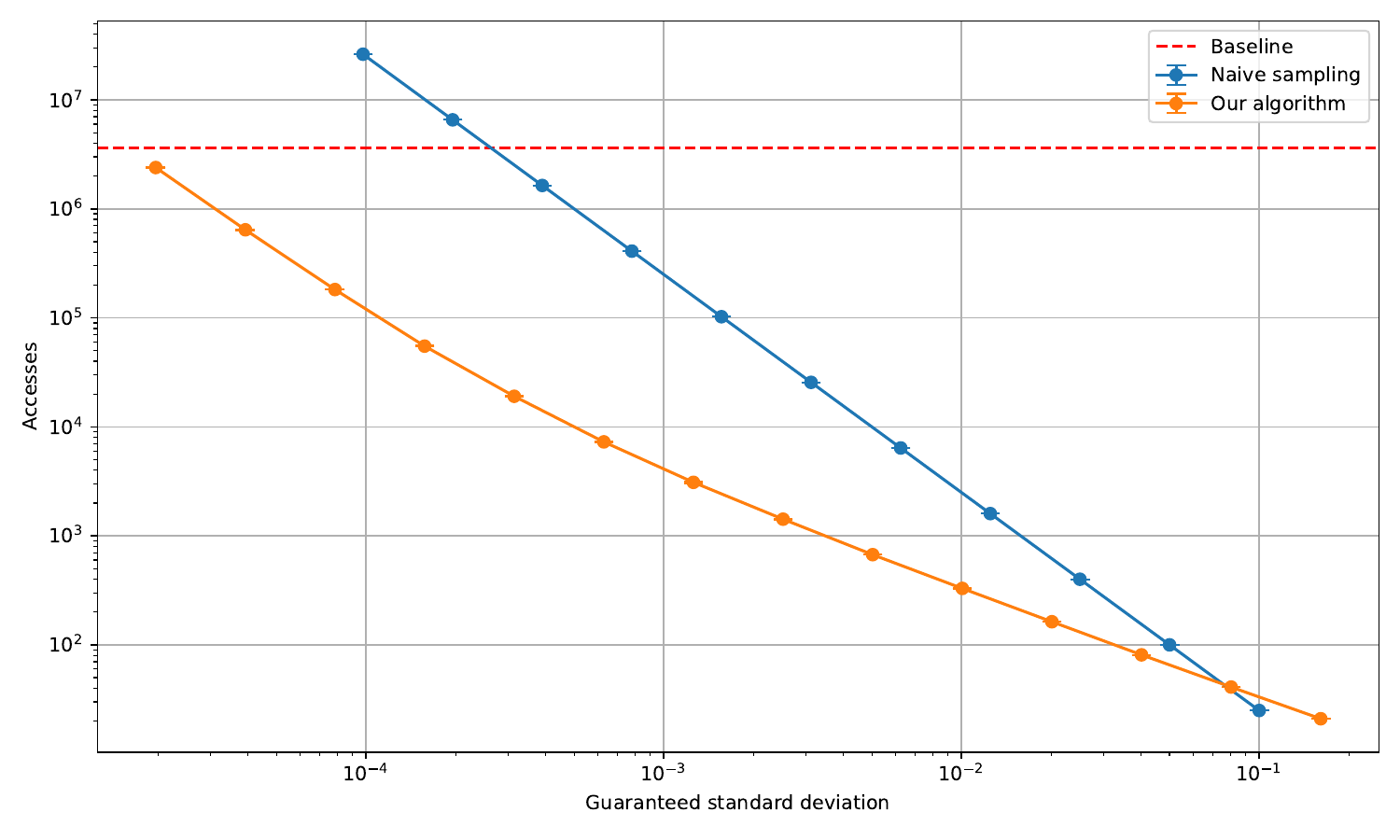}
    \caption{Estimating the frequency of the letter `e' in the corpus of English Wikipedia. }
    \label{fig:experiment}
\end{figure*}

In this section we discuss brief experiments we made with our algorithm 
We considered the problem of block mean estimation, in particular, we considered the problem of estimating the frequency of the letter `e' in the corpus of English Wikipedia \cite{germanT5_wikipedia2corpus}. We compared the following algorithms:
\begin{enumerate}
    \item Our sequential estimation algorithm \Cref{alg:counting_ones}\footnote{We in fact ran a version of the algorithm with a bit different constants and a bit more optimized guarantees from \Cref{sec:optimized}. },
    \item A naive algorithm that repeatedly samples a random letter,
    \item A baseline algorithm that reads the whole input.
\end{enumerate}
We ran the algorithm for block size $b = 4096$ (the page size of the used SSD drive) and measured two complexity measures, the actual time complexity of the algorithm and its sample complexity, i.e., the number of block retrievals. We plot these two measures in \Cref{fig:experiment} against the provable guaranteed standard deviation of the algorithm. In the case of the naive algorithm, one can see how the respective line has slope $-2$, corresponding to $\Omega(1/\eps^2)$ samples needed to guarantee $O(\eps)$ error. The curve corresponding to our algorithm has a slope around $-1$ up to $\eps \approx 10^{-3}$ or $10^{-4}$, corresponding to the term $O(1/\eps)$ dominating the complexity up to this value of $\eps$. Only for smaller values of $\eps$ the term $\sigma^2 / \eps^2$ corresponding to the slope -2 starts to dominate. 


We ran the experiments on a Lenovo ThinkPad T14 Gen 1 laptop, featuring an Intel Core i7-10610U quad-core processor, 32 gigabytes of RAM, and a Samsung NVMe solid-state drive with a capacity of 1 terabyte (MZVLB1T0HBLR-000L7). We used the ext4 file system on Ubuntu 22.04.

\paragraph{Practical variant of the mean estimation algorithm}

Here we mention a variant of the mean estimation algorithm \Cref{alg:counting_ones}. The practical issue with \Cref{alg:counting_ones} is that it ``throws away'' the samples from its first phase, they are not used towards the final estimate. We prove in \cref{sec:variants} that the algorithm retains its guarantees when samples are reused, provided that the constants in the algorithm are set to larger values that we did not try to optimize. 

\label{sec:experiments}


\bibliographystyle{plainnat}
\bibliography{literature}

\appendix

\section{Deferred Proofs from Preliminaries}
\label{sec:deferred_proofs}

We first prove \Cref{lem:expectation_of_max} that we restate here for convenience. 

\maxlemma* 
\begin{proof}[Proof of \Cref{lem:expectation_of_max}]
 We define $\ell = \log_2 k$. Let $i^* = \arg \max \|X_i\|_\ell$. Then, we have
\begin{align*}
E[\max_i X_i] \leq& \sqrt[\ell]{\E[(\max_i X_i)^\ell]} \\=& \sqrt[\ell]{\E[\max_i X_i^\ell]} \\\leq& \sqrt[\ell]{\E[\sum_{i=1}^k X_i^\ell]} \\\leq& \sqrt[\ell]{k \E[X_{i^*}^\ell]} \\=& \sqrt[\ell]{k} \|X_{i^*}\|_\ell = 2 \|X_{i^*}\|_\ell\,,
\end{align*}
where the first inequality holds by the Jensen's inequality.
\end{proof}

Next, we present the proof of \Cref{lem:prob_amplification} that we restate here for convenience. 
Recall that we say that a sequential estimation problem $P = (S, \theta, d_D)$ is \emph{monotonic} if the image of $\theta: Dist(S) \rightarrow \R$ and if $d_D(\hat\theta)$ satsifies for each $D$ that for $\hat\theta_1 \le \hat\theta_2 \le \theta(D)$ we have $d_D(\hat\theta_1) \ge d_D(\hat\theta_2)$, while for $\hat\theta_1 \ge \hat\theta_2 \ge \theta(D)$ we have $d_D(\hat\theta_1) \le d_D(\hat\theta_2)$. Recall also that $T(A,D)$ represents the expected sample complexity of the algorithm $A$ on the distribution $D$, with $\mathbf{T}$ being the random variable representing the actual number of samples performed.

\amplification*

\begin{proof}[Proof of \Cref{lem:prob_amplification}]
The first inequality is standard, we now prove the second one. We run $k = \Theta(\log \delta^{-1})$ executions of the algorithm $A$ in parallel, and once $0.9$-fraction has finished, we return their median, aborting the remaining $0.1 k$ executions. 

We now argue correctness. The problem $P$ being monotonic implies that in order to return a value with error $>\eps$, at least $(1/2-0.1)$-fraction of the $k$ executions would have to result in error $>\eps$. Since each has error $>\eps$ with probability $\leq 0.1$, by a standard application of the Chernoff's inequality the probability of this happening is at most $\delta$.

It remains to prove the concentration of the sample complexity. The argument is very similar to that for the ``median of three'' in \cite{larsen2021countsketches}. 
Using $m_i$ to denote the sample complexity of $i$-th execution, it holds that $P[m_1 \geq c \E[m_1]] \leq 1/c$ by the Markov's inequality. Therefore, by the Chernoff bound, the probability that more than $0.1$-fraction of the executions would have $m_i \geq c \E[m_1]$ is 
\begin{align*}
\leq e^{-D(0.1 | 1/c) k} &\leq e^{-0.1 \log (0.1 c) k} = (0.1 \, c)^{-\Theta(\log \delta^{-1})} = (0.1 \, c)^{-\ell}
\end{align*}
where we define $\ell = \Theta(\log \delta^{-1})$ for convenience. Let $i^*$ be the longest running execution that has not been aborted. We then have
{ \iffocs \small \fi
\begin{align*}
\|&m_{i^*}\|_{\ell-1} \\=& 
\sqrt[\ell-1]{(\ell-1) \int_0^\infty y^{\ell-2} P[m_{i*} \geq y] dy} \\\leq&
\sqrt[\ell-1]{(\ell-1) \E[m_{1}]^{\ell-1}  (\ell-1) \int_{\E[m_{1}]}^\infty  y^{\ell-2} P[m_{i^*} \geq y] dy} \\\leq&
O(\E[m_{1}]) + \sqrt[\ell-1]{(\ell-1)\int_{\E[m_{1}]}^\infty y^{\ell-2} P[m_{i^*} \geq y] dy} \\\leq&
O(\E[m_{1}]) + \sqrt[\ell-1]{(\ell-1) \int_{\E[m_{1}]}^\infty y^{\ell-2} (y/E[m_1])^{-\ell} dy} \\=&
O(\E[m_{1}]) + \sqrt[\ell-1]{(\ell-1) \E[m_{1}]^{\ell-1}} \\\leq&
O(\E[m_{1}]) = O(T(A,D))\,.
\end{align*}
}
Since there are $O(\log \delta^{-1})$ such executions running in parallel, each having complexity at most $m_{i^*}$, we finally get that
\begin{align*}
\|\mathbf{T}(A',D)\|_{\ell -1} &\leq \|m_{i^*}\|_{\ell-1} \cdot \Theta(\log \delta^{-1}) \leq O(T(A,D) \log \delta^{-1})\,.
\end{align*}
\end{proof}

We next prove \Cref{thm:io_transfer} that we restate here for convenience. 
\transtheorem*

In the proof, we use the following notation: For an input $I$ of $n$ blocks to an I/O estimation problem, we use $\fD(I)$ to denote the uniform distribution over the blocks that can serve as an input to the appropriate sequential estimation problem. 

\begin{proof}
Note that if we have a block of $B$ items, the permutation of items is irrelevant since the problem is symmetric by definition. We may thus treat a block as a multiset of elements. Throughout this proof, we then treat an element from $Dist_B(S)$ as a block of $B$ items from $S$ -- there is a natural bijection between the two,  with a block corresponding to a uniform distribution on its items and a distribution corresponding to a block which has $a$ copies of an element with probability $a/B$.

Let $A_{\eps,\delta}^{dist}$ denote an $(\eps,\delta)$-correct instance-optimal algorithm for $P_{dist}$  in the case of $(1)$. Similarly let $A_{\eps,\delta}^{I/O}$ denote an instance-optimal algorithm for $P_{I/O}$ in case of $(2)$.
We will show two reductions. First, we show a reduction $R_{I/O}$ such that $A_{\eps,\delta}^{I/O}(R_{I/O}(\cdot))$ (that is, executing $A_{\eps,\delta}^{I/O}$ on an input constructed by applying $R_{I/O}$ to $D$) is a $2\eps,2\delta$-correct algorithm for $P_{dist}$. Similarly, we show $R_{dist}$ such that $A_{\eps,\delta}^{dist}(R_{dist}(\cdot))$ is a $2\eps,2\delta$-correct algorithm for $P_{I/O}$. 
At the same time, the reductions are such that one query to $R_{I/O}(D)$ may be simulated by one sample from $D$ and a sample from $R_{dist}(I)$ can be simulated by one query to $I$.
We define the reductions next. 

\paragraph{Reduction $R_{dist}$.}
We define $R_{dist}(I)$ for input $I$ to be the distribution obtained by sampling one of $I$'s blocks at random. Clearly, one sample from $D$ can be simulated using one query to $I$. At the same time, it follows directly from the definition of a correct algorithm in the I/O estimation setting that, for any $(\eps,\delta)$-correct algorithm $A$ for $P_{dist}$, the algorithm $A(R_{dist}(\cdot))$ is $(\eps,\delta)$-correct for $P_{I/O}$.

\paragraph{Reduction $R_{I/O}$.}
We define $I = R_{I/O}(D)$ to be an input consisting of sufficiently many samples from $D$. Namely, we make so many samples that $|d_{\mathcal{D}(I)}(\hat \theta) - d_D(\hat \theta)|<\eps$ for any $\hat \theta$ except with probability $\delta$. Executing any $(\eps,\delta)$-correct algorithm for $P_{I/O}$ on $I$ will then be a $(2\eps,2\delta)$-correct algorithm for $P_{dist}$ on $D$.\\

We now wish to prove that the algorithms that we described via reductions are instance-optimal. To this end, we prove that, without loss of generality, the complexity of any algorithm $A$ on $R_{I/O}(R_{dist}(I))$ is the same up to a constant factor as on $I$, for any $I$, and similarly on $R_{dist}(R_{I/O}(D))$ the same as on $D$, for any $D$. We claim that this implies instance-optimality (up to some constant factor $c_{\eps,\delta})$). We focus on the case (1) as the argument for the other case is exactly the same. For the sake of contradiction, assume that $A_{\eps,\delta}^{dist}(R_{dist}(\cdot))$ is not an instance-optimal algorithm for $P_{I/O}$. This means we can find for any $C>0$ an algorithm $B_{\eps,\delta}^{I/O}$ and an input $I_C$ such that
\[
T(B_{\eps,\delta}^{I/O}, I_C) \leq \frac{1}{ C \cdot c_{\eps, \delta}} T(A_{\eps,\delta}^{dist}(R_{dist}(\cdot)), I_C) \,.
\]
Using this, we construct an algorithm $B_{4\eps,4\delta}^{dist}$ that is faster than $A_{\eps,\delta}^{dist}$ on some distributions by $C \cdot c_{\eps,\delta}$ for arbitrarily high $C$; since we are assuming that $P_{I/O}$ is uniformly bounded, this is already a contradiction with $A_{\eps,\delta}^{dist}$ being instance-optimal, thus concluding the proof of (1).

We define $B_{4\eps,4\delta}^{dist}(D) = B_{\eps,\delta}^{I/O}(R_{I/O}(D))$. We take the above input $I$ and let $D = R_{dist}(I)$. We then have that $B_{4\eps,4\delta}^{dist}(D)$ has the same complexity as $B_{\eps,\delta}^{I/O}$ on $I$ up to a constant (as we promised to show). This means that we have
\begin{align*}
T(B_{4\eps,4\delta}^{dist}, D) &= \Theta(T(B_{\eps,\delta}^{I/O}, I)) \leq \Theta(\frac{1}{C \cdot c_{\eps,\delta}} T(A_{\eps,\delta}^{dist}(R_{dist}(I))) = \Theta(\frac{1}{C \cdot c_{\eps,\delta}}T(A_{\eps,\delta}^{dist}, D))
\end{align*}
as we wanted to show.

\paragraph{Reductions do not change complexity.}
It remains to prove that, without loss of generality, for the above-described reductions it holds that no algorithm's expected complexity asymptotically increases if we execute it on $R_{dist}(R_{I/O}(D))$ instead on $D$, for any distribution $D$, and similarly for $R_{I/O}(R_{dist}(I))$ for any input $I$.

First, note that without loss of generality, any $(\eps,\delta)$-correct algorithm $A$ does not use more than $T(A,\cdot)/\delta$ time almost surely -- we may abort the algorithm if it does and this only increases the failure probability by $\delta$ by the Markov inequality.

We now prove the simpler case of $R_{dist}(R_{I/O}(D))$. Assuming the reduction $R_{I/O}$ samples sufficiently many samples from $D$, the probability that we sample any single block twice is $\leq \delta$. This means that the distribution $R_{dist}(R_{I/O}(D))$ is the same as $D$ up to an event of probability $\delta$. However, changing the complexity on an event with probability $\delta$ cannot asymptotically change the complexity of the algorithm, since its complexity is bounded by $T(A,\cdot)/\delta$ almost surely.

We now prove the claim for $R_{I/O}(R_{dist}(I))$. In order to do this, note that we may assume without loss of generality that the algorithm $A$ works by sampling from the input uniformly without replacement -- since the problem is by definition symmetric, we may randomly permute the input randomly; any access pattern then corresponds to sampling uniformly from the original input.
Sampling without replacement from $R_{I/O}(R_{dist}(I))$ exactly corresponds to sampling with replacement from $I$. If the input has a length at least $T(A, I)^2/\delta^3$ (the $\delta^3$ comes from the fact that the number of samples is at most $T(A,I)/\delta$, as the $T(A,I)$ bound holds only in expectation), the probability that the algorithm sees a collision when sampling is $\leq \delta$.
Like above, an event with probability $\delta$ cannot asymptotically change the complexity.
\end{proof}

\section{Variants of Mean Estimation}
\label{sec:variants}

In this section, we analyze three variants of the mean estimation algorithm \Cref{alg:counting_ones}. While \Cref{alg:counting_ones_opt} is used in our experiment, \Cref{alg:reuse} is a variant of the original algorithm that uses all samples towards the final estimate. We then optimize the constants below. We believe that this is a nice practical property that makes it worthwhile to analyze them, even though the analysis is substantially more complicated than the analysis of \Cref{alg:counting_ones}. 

\subsection{Reusing Samples}

We will now analyze \cref{alg:reuse} which is (up to constant factors) the same as \cref{alg:counting_ones} but it reuses samples from the first phase.

\begin{algorithm}
\caption{Estimate $\mu(D)$ by $\hat{\mu}$ such that $\E[(\hat{\mu} - \mu(D))^2] \leq O(\epsilon^2)$} \label{alg:reuse}
$T_1 \leftarrow 1 + \frac{10}{\eps}$ \\
$X_1, \cdots, X_{T_1} \leftarrow$ get $T_1$ samples from $D$\\
Compute $\tilde{\mu} \leftarrow \frac{1}{T_1} \sum_{i=1}^{T_1} X_i$ \label{line:empirical_variance}\\
$\tilde{\sigma}^2 \leftarrow \frac{1}{T_1-1} \sum_{i=1}^{T_1} (X_i-\tilde{\mu})^2$\\
$T_2 \leftarrow \max\left( \frac{60\tilde{\sigma}^2}{\eps^2} - T_1, 0\right)$\\
$X_{T_1+1}, \cdots, X_{T_1+T_2} \leftarrow$ get $T_2$ new samples from $D$\\ 
\Return{$\hat{\mu} = \frac{1}{T_1 + T_2} \sum_{i=1}^{T_1 + T_2} X_i$} \\
\end{algorithm}

\begin{lemma}
\label{prop:block_frequencies_ub}
\Cref{alg:reuse} returns an estimate $\hat{\mu}$ of $\mu$ such that with probability at least $1/4$ we have that 
$|\hat{\mu} - \mu(S)| \le \eps$. Moreover, the expected sample complexity of the algorithm is 
\begin{align}
\label{eq:ha}
O\left( \frac{\sigma^2}{\epsilon^2} + \frac{1}{\epsilon}\right)    
\end{align}
\end{lemma}

\begin{proof}
    First, we analyze $\tmu$. We have
    \begin{align*}
        \E[\tmu] = \frac{1}{T_1} \sum_{i = 1}^{T_1} \E[\mu(B_i)] = \mu
    \end{align*}
    and
    \begin{align*}
        \Var[\tmu] = \frac{1}{T_1^2} \sum_{i = 1}^{T_1} \E\left[(\mu - \mu(B_i))^2 \right] =  \frac{\sigma^2}{T_1} = \eps\sigma^2/10
    \end{align*}
    In particular, by Chebyshev inequality (\cref{lem:chebyshev}) we get that with probability at least $3/4$ we have
    \begin{align}
        \label{eq:am}
        \tmu = \mu \pm \sqrt{\eps} \sigma/2
    \end{align}

We need to analyze the value of $\hmu$ which can be written as 
\begin{align}
\label{eq:def}
    \hmu = \frac{T_1 \tmu + \sum_{i = T_1 + 1}^{T_1 + T_2} \mu(B_{i})}{T_1 + T_2}. 
\end{align}
We can understand the first and the second moment of $\hmu$ after fixing the random of the first phase of the algorithm. Formally, let us denote $\fF_1$ for the random bits for picking the first $T_1$ elements and let $f$ be an arbitrary instantiation of $f$; fixing $f$, we also fix the values of $\tmu, \tsigma$. We write $\E[\hmu | f], \Var[\hmu | f]$ for the conditional expectation and variance of $\hmu$, given $f$. We have for any $f$ that
\begin{align}
\label{eq:sleepy}
    \E[\hmu | f] = \frac{T_1 \tmu + T_2 \mu}{T_1 + T_2}
\end{align}
\Cref{eq:sleepy}  in particular implies
\begin{align}
\label{eq:ma}
    |\mu - \E[\hmu | f]|
    \le |\mu - \tmu|
\end{align}
Next, we compute $\Var[\hmu | f] $. Using the facts that $\Var[X +c] = \Var[X]$, $\Var[tX] = t^2\Var[X]$, and $\Var[X+Y] = \Var[X] + \Var[Y]$ for $X,Y$ independent we get 
\begin{align}
\label{eq:var}
    \Var[\hmu | f] 
    &= \frac{1}{(T_1 + T_2)^2}\cdot  \Var\left[\sum_{i = T_1 + 1}^{T_1 + T_2} \mu(B_{i}) | f\right]
    = \frac{T_2 \sigma^2}{(T_1 + T_2)^2}
\end{align}

We will now consider two cases based on which of the two terms $1/\eps, \sigma^2/\eps^2$ from \cref{eq:ha} is larger. 
\paragraph{Case $\sigma^2 \le 4\eps/9$: } We  simplify the right hand side of \cref{eq:var} by the AM-GM inequality and get
\begin{align*}
        \Var[\hmu | f] 
\le \frac{T_2 \sigma^2}{4T_1T_2} = \frac{\sigma^2}{4T_1} = \frac{\eps\sigma^2}{4}
\end{align*}
Thus, by Chebyshev (\cref{lem:chebyshev}), with probability at least $3/4$ we have 
\begin{align}
\label{eq:ham}
 \hmu = \E[\tmu | f] \pm \sqrt{\eps}\sigma   . 
\end{align}
Hence, with probability at least $1/2$, both \cref{eq:ham} and \cref{eq:am} are satisfied. Since our analysis works for any $f$, we have 
{\iffocs \small \fi
\begin{align*}
    |\hmu - \mu| 
    &\le |\hmu - \E[\hmu | f]| + |\E[\hmu | f] - \mu|  \\
    &\le |\hmu - \E[\hmu | f]| + |\mu - \tmu|    && (\text{\cref{eq:ma}})\\
    &\le \sqrt{\eps}\sigma + \sqrt{\eps}\sigma/2  
    \le \eps && \!\!\!\!\!\!\!\!\!\!\!\!\!\!\!\!\!\!\!\!\!\!\!\!\!(\text{\cref{eq:am,eq:ham}}, \sigma^2 \le \eps/9)
\end{align*}
}

\paragraph{Case $\sigma^2 > 4\eps/9$: }

First, we prove that with constant probability, we have $\hsigma^2 \ge \sigma^2/2$. To see this, first recall the standard fact that 
\begin{align}
\label{eq:standard}
    \E[\tsigma^2] = \sigma^2
\end{align}
Next, we need to bound the variance of this estimate. We have
{\iffocs \small \fi
\begin{align*}
    &\Var[\tsigma^2] =
    \\&= \E\left[ \frac{1}{(T_1-1)^2} \E\left[\sum_i \left( \tmu - \mu(B_i)\right)^2 \cdot  \sum_j \left( \tmu - \mu(B_j)\right)^2\right]  \right] - \left( \sigma^2 \right)^2\\
    &\le \frac{1}{(T_1 - 1)^2} \E\left[ \sum_i \left( \tmu - \mu(B_i) \right)^4 \right]\\
    &\le \frac{1}{(T_1 - 1)^2} \E\left[ \sum_i \left( \tmu - \mu(B_i) \right)^2 \right]
    \\&= \frac{\sigma^2}{T_1 - 1} \le \frac{\eps \sigma^2}{10} \le \sigma^2/4 
\end{align*}
}

In particular, by Chebyshev's inequality we know that with probability at least $3/4$ we have \begin{align*}
    \tsigma^2 = \sigma^2 \pm \sigma^2/2
\end{align*}
We condition on this event, note that it implies that 
\begin{align}
\label{eq:t2}
    T_2 \ge 60 \cdot \frac{\sigma^2}{2\eps^2} - 10/\eps \ge 10\sigma^2/\eps^2  
\end{align}
We also condition on \cref{eq:am} holding. We fix any instantiation $f$ of random bits in the first part of the algorithm and compute 
\begin{align*}
    \E\left[ \sum_{i = T_1 + 1}^{T_1 + T_2} \mu(B_i) | f\right] = T_2 \cdot \mu
\end{align*}
and
\begin{align*}
    \Var\left[ \sum_{i = T_1 + 1}^{T_1 + T_2} \mu(B_i) | f\right] = T_2 \cdot \sigma^2
\end{align*}
Thus, we conclude that with probability at least $3/4$ we have $\sum_{i = T_1 + 1}^{T_1 + T_2} \mu(B_i) = T_2\mu \pm 2\sqrt{T_2}\sigma$ for any choice of $f$. 
We can now finally use \cref{eq:def} to conclude that with probability at least $1/4$ we have
\begin{align*}
    \hmu &= \frac{T_1 \cdot \left( \mu \pm \sqrt{\eps}\sigma/2\right) + T_2\mu \pm 2\sqrt{T_2}\sigma}{T_1 + T_2}
    \\&= \mu \pm\left(
    \frac{\sqrt{\eps}\sigma/2}{2\sqrt{T_1T_2}} 
    +\frac{2\sigma}{\sqrt{T_2}}
    \right)\\
    &= \mu \pm \left( \eps/10 + \eps/2\right) = \mu \pm \eps
\end{align*}
where we used the AM-GM inequality for $T_1, T_2$ and plugged in \cref{eq:t2}. 

Finally, we note that the expected sample complexity of the algorithm follows from \cref{eq:standard}. 
\end{proof}

\subsection{Optimized constants}
\label{sec:optimized}
We give a slightly different variant of \Cref{alg:counting_ones} with a slightly more optimized analysis which is used in \Cref{sec:experiments}. 

\begin{algorithm}
$T_1 \leftarrow \frac{1}{\eps}$ \\
$X_1, \cdots, X_{T_1} \leftarrow$ get $T_1$ samples from $D$\\
Compute $\tilde{\mu} \leftarrow \frac{1}{T_1} \sum_{i=1}^{T_1} X_i$ \\
$\tilde{\sigma}^2 \leftarrow \frac{1}{T_1-1} \sum_{i=1}^{T_1} (X_i-\tilde{\mu})^2$\\
$T_2 \leftarrow \frac{1}{\eps} + 5\frac{\tilde{\sigma}^2}{\eps^2}$\\
$X_{T_1+1}, \cdots, X_{T_1+T_2} \leftarrow$ get $T_2$ new samples from $D$\\ 
\Return{$\hat{\mu} = \frac{1}{T_2} \sum_{i=T_1+1}^{T_2} X_i$} \\
\caption{Estimate $\mu(D)$ by $\hat{\mu}$ such that $\E[(\hat{\mu} - \mu(D))^2] \leq O(\epsilon^2)$} \label{alg:counting_ones_opt}
\end{algorithm}


\begin{theorem} \label{thm:mean_estimation_ub_opt}
Consider any distribution supported on $[0,1]$ with mean $\mu$ and variance $\sigma^2$. \Cref{alg:counting_ones_opt} returns an unbiased estimate $\hat{\mu}$ of $\mu$ with standard deviation of $\sqrt{\E[(\hat{\mu}-\mu)^2]} \leq \sqrt{2.58} \cdot \epsilon$. Its expected sample complexity is $5\frac{\sigma^2}{\epsilon^2} + \frac{2}{\epsilon}$.
\end{theorem}

In the following proof, we use $\mu, \sigma^2$ to represent the true mean and variance of the distribution $D$. We also use $\E_1, \E_2$ whenever we want to stress that the randomness is only over the first or the second phase of the algorithm. 

\begin{proof}

We start by bounding the expected query complexity.
The algorithm performs $\hat\sigma^2/\eps^2 + 2/\eps$ samples. As $\hat\sigma^2$ is an unbiased estimator of $\sigma^2$, we conclude that the expected sample complexity is $\sigma^2/\eps^2 + 2/\eps$, as needed. 

We next prove correctness. We first argue that the estimate $\hat\mu$ is unbiased. Note whatever values $x_1, \dots, x_{T_1}$ were sampled in the first phase, we have 
\[
\E[X_i | X_1 = x_1, \dots, X_{T_1} = x_{T_1}] = \mu
\]
for each $T_1 + 1 \le i \le T_1 + T_2$. Hence, $\E[\hat{\mu}|X_1 = x_1, \dots, X_{T_1} = x_{T_1}] = \mu$ and by the law of total expectation also $\E[\hat\mu] = \mu$. 


We now bound the variance. As $\E[\hat{\mu}]=\mu$, we have
\[\Var[\hat{\mu}] = \E[(\hat{\mu}-\mu)^2] = \E_{1}[\E_{2}[(\hat{\mu}-\mu)^2 | T_2]],\]
where the last line follows from the law of total expectation.
Conditioned on $T_2 = t_2$%
, we have that $\E_2[(\hat{\mu}-\mu)^2|T_2 = t_2] = \frac{\sigma^2}{t_2}$, so we can write
\begin{equation} \label{eq:mean-ub-eq1_2}
\Var[\hat{\mu}] = \E_1\left[\frac{\sigma^2}{T_2}\right] = \E_1\left[\frac{\sigma^2}{\frac{1}{\eps} + 5\frac{\tilde{\sigma}^2}{\eps^2}}\right].    
\end{equation}

Next, we need to bound $\tilde{\sigma}^2$. Using the formula from \Cref{lem:variance_of_variance}, we have 
{\iffocs \small \fi
\begin{align*}&Var[\tilde{\sigma}^2] \\&= \frac{1}{T_1} \E_{X \sim D} [(X-\mu)^4] - \frac{T_1-3}{T_1(T_1-1)} \cdot \E_{X \sim D} [(X-\mu)^2]^2 \\&\le \frac{1}{T_1} \E_{X \sim D} [(X-\mu)^4].\end{align*}
}
Since we always have $|X-\mu| \le 1$, we have $(X-\mu)^4 \le (X-\mu)^2,$ and we can further bound 
\[
\Var[\tilde{\sigma}^2]  \le \frac{1}{T_1} \cdot \E_{X \sim D} [(X-\mu)^2] = \eps \cdot \sigma^2.
\]

Therefore, by Chebyshev's inequality \Cref{lem:chebyshev} and using that $\tilde{\sigma}^2$ is unbiased by \Cref{lem:variance_of_variance}, it holds that 
\[\P\left(\tilde{\sigma}^2 < \frac{\sigma^2}{4}\right) \le \frac{\eps \cdot \sigma^2}{(3/4 \cdot \sigma^2)^2} = \frac{16 \eps}{9\sigma^2}.\]

Let $\mathcal{E}_1$ be the event that $\tilde{\sigma}^2 < \frac{\sigma^2}{4}$. 
In this event, we have 
$$\frac{\sigma^2}{(1/\eps)+5(\tilde{\sigma}^2/\eps^2)} \le \sigma^2 \cdot \eps.$$ Otherwise, 
$$\frac{\sigma^2}{(1/\eps)+5(\tilde{\sigma}^2/\eps^2)} \le \frac{\sigma^2}{5(\sigma^2/4)/\eps^2} \le 4/5 \cdot \eps^2. $$
Therefore,
\begin{align} \label{eq:mean-ub-eq2_2}
    \E\left[\frac{\sigma^2}{\frac{1}{\eps}+\frac{5\tilde{\sigma}^2}{\eps^2}}\right] \nonumber
    &\le \P(\mathcal{E}_1) \cdot \sigma^2 \eps + (1-\P(\mathcal{E}_1)) \cdot 4/5 \cdot \eps^2\\&\le \frac{16\eps}{9\sigma^2} \cdot \sigma^2 \eps + 4/5 \eps^2 \le 2.58 \eps^2.
\end{align}
Combining Equations \eqref{eq:mean-ub-eq1_2} and \eqref{eq:mean-ub-eq2_2} finishes the proof. 
\end{proof}

\section{Instance Optimality Gap for Median}
\label{sec:gap}

Recall that in \Cref{thm:quantile_main} we have presented an algorithm for quantile estimation that is instance-optimal up to a (very small) factor of 
\[
1 + \frac{\log \left( \delta^{-1} \log \eps^{-1} \right)}{\log \delta^{-1}}
\]
which is of order $O(\log\log \eps^{-1})$ for constant $\delta$. 

In this section, we show that this small gap is necessary in the sense that this factor in \Cref{thm:quantile_main} cannot be improved. 
We will focus only on estimating the median. Formally, we define the problem \median as \quantile($1/2$). 
We focus on the following simple setup: The input distribution is either a Bernoulli distribution $B(1/2 + \bar\eps)$ or $B(1/2 - \bar\eps)$. 
Our task is to estimate the median up to additive $\eps$. 
In particular, for $\bar\eps > \eps$, on input $B(1/2 + \bar\eps)$ we have to return \texttt{1} and on input $B(1/2 - \bar\eps)$ we have to return \texttt{0}. 

On one hand, we will show that for each $D = B(1/2 + \bar\eps)$ there exists an $(\eps,\delta)$-estimator $A_D$ that uses only $O(1/\bar\eps^2)$ samples on $D$. On the other hand, we will show that any $(\eps, \delta)$-estimator has to make at least $\Omega(1/\bar\eps^2 \cdot \log\log \eps^{-1})$ samples for at least one $\bar\eps > \eps$, showing the instance-optimality gap. The theorem is stated formally next. 

\begin{theorem}
\label{thm:gap_median}
    The instance optimality gap in the \median problem  is 
    \begin{align}
    \label{eq:factor}
       \Omega\left(1 + \frac{\log \left( \delta^{-1} \log \eps^{-1} \right)}{\log \delta^{-1}}\right).  
    \end{align}    
    That is, for every $\eps, \delta > 0$ there is a finite set of input distributions $\fD$, a set of $(\eps, \delta)$-estimators $\{ A_D, D \in \fD\}$, and a constant $C$ such that the following is true. 
    For every $(\eps, \delta)$-estimator $A$, there exists at least one $D \in \fD$ such that:
    $$
    T(A, D)  \ge  C \cdot \left( 1 + \frac{\log \left( \delta^{-1} \log \eps^{-1} \right)}{\log \delta^{-1}} \right) \cdot T(A_D, D). 
    $$
\end{theorem}

The theorem is proven in the rest of this section. 

\paragraph{Instances with fast instance-specific algorithms}
We start by fixing any positive constants $\eps, \delta > 0$ that are sufficiently small. We will also assume that 
\begin{align}
    \label{eq:eps_delta}
    \log \eps^{-1} > \delta^{-1}
\end{align}
since otherwise \cref{eq:factor} reduces to a constant and there is nothing to prove. 

We start by formally specifying the set of input distributions $\fD$ of interest and show a fast instance-specific algorithm for each such distribution. 
Given $\eps > 0$, we will consider the following set of distributions. Let 
\begin{align}
    \label{def:ei}
    \eps_i := \left( \frac{1}{100 \log^3 \eps^{-1}} \right)^i
\end{align}
and consider the sequence $\eps_1, \eps_2, \dots, \eps_\ell$ for $\ell = {O(\log \eps^{-1} / \log\log \eps^{-1})}$ chosen so that for each $i \le \ell$ we have $\eps_i > \eps$. 
For each such $\eps_i$, we consider the following two Bernoulli distributions: a Bernoulli distribution that returns $0$ with probability $1/2 + \eps_i$ that we denote $B(1/2 + \eps_i)$, and the analogously defined distribution $B(1/2 - \eps_i)$. 
We let $\fD$ to be the set of $2t$ of these Bernoulli distributions. 

Next, we need to define the set of algorithms $A_D$ for each $D \in \fD$. 

\begin{claim}
For each distribution $D = B(1/2 + \eps_i)$ or $D = B(1/2 - \eps_i)$, $ D \in \fD$, there is an $(\eps, \delta)$-estimator $A_D$ of \median such that $T(A_D, D) = O(1/\eps_i^2 \cdot \log \delta^{-1})$. 
\end{claim}
\begin{proof}
We consider only the case $D = B(1/2 + \eps_i)$. The algorithm $A_D$ runs two parallel algorithms $A_{correct}$ and $A_{test}$, meaning that it always samples two independent samples, feeds one sample to each of the two algorithms, and waits until one of them finishes and returns an estimate that $A_D$ returns, too.  

The algorithm $A_{correct}$ is any $(\eps, \delta/2)$-estimator of \median. The algorithm $A_{test}$ works as follows. 
Choose a large enough $C_0$. In the $j$-th phase, the algorithm samples $N_j = C_0^j / \eps_i^2 \cdot \log \delta^{-1}$ fresh samples and computes the frequency of ones in them. 
If it is at least $1/2 + \eps_i/2$, the algorithm returns \texttt{1}. Otherwise, the algorithm $A_{test}$ proceeds to the next phase; in particular, it never returns \texttt{0}. This finishes the description of the algorithm. 

To see that $A_D$ is an $(\eps, \delta)$-estimator, consider any input distribution $D_0$. 
If the probability of sampling \texttt{1} under $D_0$ is at least $1/2$, $A_{test}$ never returns a wrong answer and thus $A_D$ is an $(\eps, \delta/2)$-estimator. Otherwise, each phase of the algorithm $A_{test}$ has some probability of returning a wrong answer since the observed frequency of \texttt{1} can be more than $1/2 + \eps_i/2$ even though the probability of sampling it is less than $1/2$. In particular, the probability that $A_{test}$ returns \texttt{1} in the $j$-th phase can be bounded by Chernoff bound (\Cref{lem:chernoff}) as 
{\iffocs \small \fi
\begin{align*}
    2\exp\left( -2 \cdot \left( \frac{\eps}{2} \cdot N_j \right)^2 / N_j \right) \le \exp \left( - C_0^j/ 10 \cdot \log \delta^{-1}  \right)
\end{align*}
}
Summing up this upper-bound over all $j \ge 0$, we conclude that the overall failure probability is at most $\delta/2$. We union bound with the possibility of failure of $A_{correct}$ to conclude that $A_D$ fails with probability at most $\delta$. 

To see that $T(A_D, D) = O(1/\eps_i^2)$, we notice that while the frequency of \texttt{1} in $D$ is $1/2 + \eps_i$, the algorithm $A_{test}$ stops once the observed frequency is $1/2 + \eps_i/2$; thus we can again use Chernoff bound to compute that the probability that $A_D$ does not finish after phase $j$ is at most $\exp \left( - C_0^j /10\right)$. 
The expected sample complexity of $A_D$ is thus at most
\begin{align*}
     O(1/\eps_i^2) \cdot \sum_{j \ge 0} C_0^j \cdot \exp \left( - C_0^j /10\right) =     O(1/\eps_i^2)
\end{align*}
as needed. 
\end{proof}

\paragraph{Lower bound}
We next need to prove that for every $(\eps, \delta)$-estimator $A$, there is at least one input $D \in \fD$ such that if that input is $B(1/2-\eps_i)$ or $B(1/2+\eps_i)$, then $T(A, D) = \Omega( 1/\eps_i^2 \cdot \log \left( \delta^{-1} \log \eps^{-1} \right))$. 
We will next for contradiction assume that this is not the case. 
That is, we will assume that for any constant $C > 0$  there is an algorithm $A_0$ such that for every $D = B(1/2 \pm \eps_i) \in \fD$, we have
\begin{align}
    \label{eq:contradiction}
T(A_0, D) < \frac{1}{C} \cdot \frac{1}{\eps_i^2} \cdot \log \left( \delta^{-1} \log \eps^{-1} \right)   
\end{align}

We will now show how to derive a contradiction from the above statement, provided that $\eps, \delta$ are small enough. 

We split the run of $A_0$ into phases, where the $i$-th phase finishes after $A_0$ samples $t_i$ elements, where we define $t_i$ as 
\begin{align}
    \label{eq:phase_len}
    t_i = \frac{100}{C} \cdot \frac{1}{\eps_i^2} \cdot \log \left( \delta^{-1} \log \eps^{-1} \right). 
\end{align}

Next, we use Hellinger distance to formalize that $A_0$ cannot distinguish $B(1/2 - \eps_i), B(1/2)$, and $B(1/2 + \eps_i)$ until the $i$-th phase, since the distributions can be coupled with very small error. Importantly, even after the $i$-th phase, there is a non-negligible coupling between the distributions. This is formalized in the next claim. 

\begin{claim}
    \label{cl:coupling_small}
    Fix any $i$ with $1 \le i \le \ell$. For any two different distribution $D_1, D_2 \in \{ B(1/2 - \eps_i), B(1/2), B(1/2 + \eps_i) \} $, we have 
    \begin{align}
        \label{eq:hsmall}
        H^2(D_1^{\otimes t_{i-1}}, D_2^{\otimes t_{i-1}}) \le 10^{-6}/\log^2 \eps^{-1}
    \end{align}
    and
    \begin{align}
        \label{eq:hsmall2}
        H^2(D_1^{\otimes (t_i - t_{i-1})}, D_2^{\otimes (t_i - t_{i-1})}) \le 1 - \left( \frac{\delta}{\log\eps^{-1}} \right)^{0.1}
    \end{align}
    In particular, by \Cref{fact:hellinger_implies_l1} we have
    \begin{align}
        \label{eq:tvsmall}
        d_{TV}(D_1^{\otimes t_{i-1}}, D_2^{\otimes t_{i-1}}) \le 0.01 / \log \eps^{-1}
    \end{align}
    and
    \begin{align}
        \label{eq:tvsmall2}
        d_{TV}(D_1^{\otimes (t_i - t_{i-1})}, D_2^{\otimes (t_i - t_{i-1})}) \le 1 - \frac{\delta}{\log \eps^{-1}}
    \end{align}

\end{claim}
\begin{proof}
    We will prove the results just for $D_1 = B(1/2 - \eps_i)$ and $D_2 = B(1/2 + \eps_i)$. 
    We start with \Cref{eq:hsmall}. 
    Using \Cref{def:hellinger} and \Cref{fact:sqrt}, we have
    \begin{align*}
        H^2(D_1, D_2)
        &= \left( \sqrt{1/2 + \eps_i} - \sqrt{1/2 - \eps_i} \right)^2\\
        &= \frac12 \left( \sqrt{1+2\eps_i} - \sqrt{1 - 2 \eps_i}\right)^2\\
        & \le \frac12 \left( (1+\eps_i) - (1 - \eps - 2\eps_i^2)\right)
        = \eps_i^2
    \end{align*}
    Next, note that
    \begin{align*}
        t_{i-1} &= \frac{100}{C} \log \left( \delta^{-1} \log \eps^{-1} \right) \cdot \frac{1}{\eps_{i-1}^2}
        \\&=\frac{100}{C} \log \left( \delta^{-1} \log \eps^{-1} \right) \cdot \frac{1}{\eps_{i}^2} \cdot \frac{1}{(100 \log^3 \eps^{-1})^2}
        \\&\le \frac{1}{10^{10}\eps_i^2 \log^2 \eps^{-1}}
    \end{align*}
    where we used that $C$ is large enough and \cref{eq:eps_delta}. 
    Thus, using this bound on $t_{i-1}$ and the product properties of Hellinger distribution of \Cref{fact:hellinger_product}, we get
    \begin{align*}
        H^2(D_1^{\otimes t_{i-1}}, D_2^{\otimes t_{i-1}}) =
        &= 1 - (1 - H^2(D_1, D_2))^{t_{i-1}}
        \\& \le 1 - (1 - \eps_i^2)^{1/(10^{10}\eps_i^2\log^2 \eps^{-1}) } \\
        &\le 10^{-6} / \log^2 \eps^{-1}
    \end{align*}
    as needed. 

    We continue similarly with \Cref{eq:hsmall2}. We have 
    \begin{align*}
        H^2(&D_1^{\otimes (t_i - t_{i-1})}, D_2^{\otimes (t_i - t_{i-1})})
        \\&= 1 - (1 - H^2(D_1, D_2))^{t_i - t_{i-1}}
        \\&\le 1 - (1 - \eps_i^2)^{\frac{100\log \left( \delta^{-1} \log \eps^{-1} \right)}{C} \frac{1}{\eps_i^2}}\\
        &\le 1 - \exp\left( - 1000/ C \cdot \log \left( \delta^{-1} \log \eps^{-1} \right)  \right)
        \\&\le 1 - \left( \frac{\delta}{\log \eps^{-1}}\right)^{0.1}
    \end{align*}
    using that $C$ is large enough. 
    
    To prove equations \Cref{eq:tvsmall,eq:tvsmall2}, we rewrite the second inequality of \Cref{fact:hellinger_implies_l1} as 
    $$
    d_{TV}(p,q) \le \sqrt{ 1 - (1 - H^2(p,q))^2} 
    $$
    and in particular, write
    \begin{align*}
        d_{TV}(D_1^{\otimes t_{i-1}}, D_2^{\otimes t_{i-1}})
        &\le \sqrt{1 - (1 - 10^{-6} / \log^2 \eps^{-1})^2} \\&\le 0.01 / \log \eps^{-1}
    \end{align*}
    and
    \begin{align*}
        d_{TV}(&D_1^{\otimes (t_i - t_{i-1})}, D_2^{\otimes ((t_i - t_{i-1}))}) \leq
        \\&\le \sqrt{1 - \left(1 - \left(1 - \frac{1}{\left( \frac{\delta}{\log \eps^{-1}}\right)^{0.1}}\right)\right)^2}\\ 
        &= \sqrt{1 - \frac{\delta^{0.2}}{\log^{0.2}\eps^{-1}} } 
    \le  1 - \frac{\delta}{\log \eps^{-1}}
    \end{align*}
    \end{proof}

We will prove by induction that $A_0$ has to always have a nonnegligible chance of terminating in each phase when run on $B(1/2)$. 
We use the following notation. Let $S_i$ be the event that $A_0$ stops in the $i$-th phase and similarly, $S_{\le i}$ the event that it stops before the end of the $i$-th phase. 
We also use the notation $\P_p$ to denote the probability of an event if the input distribution is $B(p)$. 
We will next bound probabilities of various stopping events when $A_0$ is run on $B(1/2 + \eps_i)$ or $B(1/2 - \eps_i)$. 

\begin{claim}
\label{cl:ps}
For every $0 \le i \le \ell$ we have 
\begin{enumerate}
    \item $\P_{1/2 + \eps_i}(S_{\le i}) \ge 0.99$
    \item $\P_{1/2 + \eps_i} ( S_{\le i - 1} ) \le 0.4$
    \item $\P_{1/2 + \eps_i}(S_i ) \ge 0.5$ 
    \item Define $S' \subseteq \neg S_{\le i-1}$ to be the set of all elementary events $e$ for which $\P_{1/2 + \eps_i}(S_i | e) \ge 0.1$. We have $\P_{1/2 + \eps_i}(S' ) \ge 0.4 $. 
\end{enumerate}
The same holds if we replace $1/2 +\eps_i$ by $1/2 - \eps_i$. 
\end{claim}
\begin{proof}
We start with the first item. We note that for the expected sample complexity of $A_0$ we have
\begin{align*}
    T(A_0, B(1/2 + \eps_i))
    \ge t_i \cdot \left( 1 - \P_{1/2 + \eps_i}(S_{\le i}) \right)
\end{align*}
Plugging in our assumption from \Cref{eq:contradiction} and the definition of $t_i$ from \Cref{eq:phase_len}, we thus conclude that
\begin{align*}
    \frac{1}{C} \cdot &\frac{1}{\eps_i^2} \cdot \log \left( \delta^{-1} \log \eps^{-1} \right) \geq 
    \\&\ge \frac{100}{C} \cdot \frac{1}{\eps_i^2} \cdot \log \left( \delta^{-1} \log \eps^{-1} \right) \cdot \left( 1 - \P_{1/2 + \eps_i}(S_{\le i}) \right)
\end{align*}
Hence,
\begin{align}
    \label{eq:plarge}
    \P_{1/2 + \eps_i}(S_{\le i}) \ge 0.99. 
\end{align}
We similarly have $\P_{1/2 - \eps_i}(S_{\le i}) \ge 0.99$ which proves the first item. 

We continue with the second inequality. For a contradiction, assume that 
$$
\P_{1/2 + \eps_i}(S_{\le i - 1}) > 0.4. 
$$
Hence, since $A_0$ is $(\eps, \delta)$-estimator on $B(1/2 + \eps_i)$ and since we can assume $\delta < 0.1$, we have
$$
\P_{1/2 + \eps_i}(S_{\le i-1} \cap \text{ $A_0$ outputs 1}) > 0.3. 
$$
Finally, \Cref{cl:coupling_small} provides a coupling between $B(1/2 + \eps)^{t_{i-1}}$ and $B(1/2 - \eps)^{t_{i-1}}$ up to $0.01$ error, thus
$$
\P_{1/2 - \eps_i}(S_{\le i-1} \cap \text{ $A_0$ outputs 1}) > 0.29 > 0.1, 
$$
a contradiction with $A_0$ be $(\eps, 0.1)$-estimator on $B(1/2 - \eps_i)$. The same holds if we swap $1/2 + \eps_i$ with $1/2 - \eps_i$. 

The third item follows from the fact that $\P_{1/2 + \eps_i}(S_i ) = \P_{1/2 + \eps_i}(S_{\le i} ) - \P_{1/2 + \eps_i}(S_{\le i-1} ) > 0.5 $ where we used the previous two inequalities. 

Finally, the last item is proven by the following argument. 
Notice that by the definition of $S'$, we can write
\begin{align*}
    \P_{1/2 + \eps_i}(S_{i} )
    &= \P_{1/2 + \eps_i}(S_{i} \cap S') + \sum_{e \not\in S'} \P(e) \cdot \P(S_i | e)
    \\&\le \P_{1/2 + \eps_i}(S_{i} \cap S') + 0.1
\end{align*}
On the other hand, notice that $\P_{1/2 + \eps_i}(S_{i} ) \ge 0.5$ by the third item. Putting the two inequalities together, we conclude that $\P_{1/2 + \eps_i}(S' ) \ge \P_{1/2 + \eps_i}(S_{i} \cap S') \ge 0.4$, as needed. 

\end{proof}

We can now prove the following proposition: 
\begin{proposition}
\label{prop:loglogmain}
    For every $0 \le i \le \ell$, we have that 
    \begin{align*}
        \P_{1/2}(S_i) \ge \frac{\delta}{100 \sqrt{\log \eps^{-1}}}
    \end{align*}
\end{proposition}

\begin{proof}
    We consider the following coupling between $B(1/2)^{\otimes t_i}$ and $B(1/2 + \eps)^{\otimes t_i}$. First, we couple the first $t_{i-1}$ samples via item 3 in \Cref{cl:coupling_small} so that the coupling fails on an event of probability at most $0.01$. 
    Next, for any elementary event $e \in B(1/2)^{\otimes t_{i-1}}$ such that $A_0$ run on $e$ does not terminate in the first $t_{i-1}$ steps, we couple the following $t_i - t_{i-1}$ steps using item 4 in \Cref{cl:coupling_small} so that the coupling works on an event of probability at least $\delta / \log \eps^{-1}$. 

    To compute $\P_{1/2}(S_i)$, 
    we first note that for every elementary event $e \not\in S_{\le i-1}$, we have 
    \begin{align}
    \label{eq:b}
    \P_{1/2}(S_i | e) > \P_{1/2 + \eps_i}(S_i | e)  \cdot \frac{\delta}{\sqrt{\log \eps^{-1}}}    
    \end{align}
    by our coupling of the last $t_i - t_{i-1}$ steps. Next, we define the set $S'$ as in \cref{cl:ps} and write  

    \iffocs
   {\footnotesize
    \begin{align*}
    &\P_{1/2}(S_i) \geq \\
    &\ge    \P_{1/2}(S' \cap S_i)\\
    &= \sum_{e \in S'} \P_{1/2}(e) \cdot \P_{1/2}(S_i | e) \\
    &\ge \sum_{e \in S'} \P_{1/2}(e) \cdot \P_{1/2 + \eps_i}(S_i | e) \cdot \frac{\delta}{\sqrt{\log \eps^{-1}}} && \text{(\Cref{eq:b})}\\
    &\ge \frac{0.1 \delta}{\sqrt{\log \eps^{-1}}} \cdot \sum_{e \in S'} \P_{1/2}(e)&&\!\!\!\!\text{(Definition of $S'$)}\\
    &= \frac{0.1\delta}{\sqrt{\log \eps^{-1}}} \cdot \P_{1/2}(S')\\
    &\ge \frac{0.1 \delta \cdot (\P_{1/2 + \eps_i}(S') - 0.01)}{\sqrt{\log \eps^{-1}}} && \!\!\!\!\!\!\!\!\!\!\!\!\!\!\!\!\!\!\!\!\!\!\!\!\!\!\!\!\!\!\!\!\!\!\!\!\!\!\!\text{(second inequality of \Cref{cl:ps})}\\
    &\ge \frac{\delta}{100\sqrt{\log \eps^{-1}}} && \!\!\!\!\!\!\!\!\!\!\text{(\Cref{cl:ps}, item 4)}
    \end{align*}
    }
  \else
    \begin{align*}
    \P_{1/2}(S_i) 
    &\ge    \P_{1/2}(S' \cap S_i)\\
    &= \sum_{e \in S'} \P_{1/2}(e) \cdot \P_{1/2}(S_i | e) \\
    &\ge \sum_{e \in S'} \P_{1/2}(e) \cdot \P_{1/2 + \eps_i}(S_i | e) \cdot \frac{\delta}{\sqrt{\log \eps^{-1}}} && \text{\Cref{eq:b}}\\
    &\ge \frac{0.1 \delta}{\sqrt{\log \eps^{-1}}} \cdot \sum_{e \in S'} \P_{1/2}(e)&&\text{Definition of $S'$}\\
    &= \frac{0.1\delta}{\sqrt{\log \eps^{-1}}} \cdot \P_{1/2}(S')\\
    &\ge \frac{0.1 \delta \cdot (\P_{1/2 + \eps_i}(S') - 0.01)}{\sqrt{\log \eps^{-1}}} && \text{second inequality of \Cref{cl:ps}}\\
    &\ge \frac{\delta}{100\sqrt{\log \eps^{-1}}} && \text{\Cref{cl:ps}, item 4}
    \end{align*}
    \fi

    as needed. 
\end{proof}

We can now finish the proof by reaching a contradiction with the fact that $A$ is correct with error at most $\delta$. 

\begin{claim}
\label{cl:cont}
    The algorithm $A_0$ returns a wrong answer on either $B(1/2 + \eps_\ell)$ or $B(1/2 - \eps_\ell)$ with probability at least $2\delta$. 
\end{claim}
\begin{proof}
    Recall that $\ell = \Omega(\log \eps^{-1} / \log\log\eps^{-1})$, thus for $\eps$ small enough we have by \Cref{prop:loglogmain} that 
    \begin{align}
    \label{eq:whyme}
    \P_{1/2}(S_1 \cup \dots \cup S_{\ell-1}) > 100 \delta.        
    \end{align}
    Using \Cref{eq:tvsmall} of \Cref{cl:coupling_small}, we can construct a coupling of $B(1/2)^{\otimes t_{\ell-1}}, B(1/2 + \eps_t)^{\otimes t_{\ell-1}}, B(1/2 - \eps_t)^{\otimes t_{\ell-1}} $ on an event of probability $1 - 0.02/\log \eps^{-1}$. 
   Using \cref{eq:eps_delta}, we conclude that our coupling fails on an event of probability at most $\delta$. 
    In particular, consider the event $E$ of $A_0$ terminating before the beginning of the $\ell$-th phase. The probability of $E$ under our coupling is at least $100\delta - \delta = 99\delta$ for both $B(1/2 + \eps_{\ell-1})$ and $B(1/2 -  \eps_{\ell-1})$ by \cref{eq:whyme} and the coupling. But this means that on at least one of the two distributions, $A_0$ returns a wrong answer with probability at least $99\delta/2 > 2\delta$.     
\end{proof}

\end{document}

\end{document}

%% file: macros.tex
\pdfoutput=1 

\newif\ifconference
\conferencefalse

\usepackage[utf8]{inputenc}
\usepackage[english]{babel}
\usepackage{amsmath,amsfonts}
\usepackage{amssymb}
\usepackage[showonlyrefs]{mathtools}
\usepackage{amsthm}
\usepackage{moresize}
\usepackage[a4paper, portrait, left=1in, right=1in, top=1in, bottom=1in]{geometry}
\usepackage[numbers,sectionbib,longnamesfirst]{natbib}
\usepackage[hyperindex,breaklinks,bookmarks]{hyperref}
\hypersetup{bookmarksdepth=2}
\usepackage[shortlabels]{enumitem}
\iffocs \else
\usepackage[toc]{appendix}
\fi
\usepackage{microtype}
\usepackage[ruled,linesnumbered]{algorithm2e}
\usepackage[capitalize]{cleveref}
\usepackage{xcolor}
\usepackage[textsize=tiny]{todonotes}
\usepackage{lineno}
\usepackage{xspace}

\newcommand{\jakub}[1]{}
\newcommand{\vasek}[1]{}

\SetKwFor{RepTimes}{repeat}{times}{end}

\newtheorem{theorem}{Theorem}[section]
\newtheorem{corollary}[theorem]{Corollary}
\newtheorem{definition}[theorem]{Definition}
\newtheorem{lemma}[theorem]{Lemma}
\newtheorem{proposition}[theorem]{Proposition}

\newtheorem{fact}[theorem]{Fact}
\newtheorem{claim}[theorem]{Claim}
\newtheorem{result}[theorem]{Result}

\theoremstyle{remark}

\usepackage{thm-restate}

\renewcommand{\tilde}{\widetilde}
\renewcommand{\hat}{\widehat}

\renewcommand{\P}{\mathrm{\mathbf{P}}}
\newcommand{\E}{\mathrm{\mathbf{E}}}

\newcommand{\Var}{\mathrm{\mathbf{Var}}}

\newcommand{\eps}{\epsilon}
\newcommand{\fA}{\mathcal{A}}

\newcommand{\fD}{\mathcal{D}}

\newcommand{\fF}{\mathcal{F}}

\newcommand{\fI}{\mathcal{I}}

\newcommand{\fS}{\mathcal{S}}

\renewcommand{\phi}{\varphi}

\newcommand{\R}{\mathbb{R}}

\newcommand{\tmu}{\tilde{\mu}}
\newcommand{\tsigma}{\tilde{\sigma}}
\newcommand{\hmu}{\hat{\mu}}
\newcommand{\hsigma}{\hat{\sigma}}
\newcommand{\tN}{\tilde{N}}
\newcommand{\hN}{\hat{N}}
\newcommand{\tq}{\tilde{q}}

\makeatletter
\newcommand{\algorithmstyle}[1]{\renewcommand{\algocf@style}{#1}}
\makeatother

\newcommand{\bmean}{\textsc{Block Mean}\xspace}
\newcommand{\stepmean}{\textsc{Step-Mixture Mean}\xspace}
\newcommand{\mean}{\textsc{Mean}\xspace}
\newcommand{\freq}{\textsc{Frequency}\xspace}
\newcommand{\quantile}{\textsc{Quantile}\xspace}
\newcommand{\median}{\textsc{Median}\xspace}
\newcommand{\mquantile}{\textsc{Mixture Quantile}\xspace}
\newcommand{\distks}{\textsc{Distribution learning (KS)}\xspace}
\newcommand{\bdistks}{\textsc{Block Distribution learning (KS)}\xspace}

\newcommand{\mks}{\textsc{Mixture learning (KS)}\xspace}
\newcommand{\bfreq}{\textsc{Block Frequency}\xspace}
\newcommand{\stepfreq}{\textsc{Step-Mixture Frequency}\xspace}
\newcommand{\minfty}{\textsc{Mixture learning ($\ell_\infty$)}\xspace}
\newcommand{\stepinfty}{\textsc{Step-Mixture learning ($\ell_\infty$)}\xspace}
\newcommand{\distinfty}{\textsc{Distribution learning ($\ell_\infty$)}\xspace}
\newcommand{\bdistinfty}{\textsc{Block Distribution learning ($\ell_\infty$)}\xspace}
\newcommand{\bquantile}{\textsc{Block Quantile}\xspace}
\newcommand{\stepquantile}{\textsc{Step-Mixture Quantile}\xspace}